\documentclass{article}
\usepackage[utf8]{inputenc}
\usepackage{amsmath, amsthm, mathtools, amsfonts, graphicx, hyperref,setspace, xparse}
\usepackage{breqn}

\usepackage{biblatex}
\theoremstyle{plain}
\newtheorem{thm}{Theorem}[section]
\newtheorem{lem}[thm]{Lemma}
\newtheorem{prop}[thm]{Proposition}
\newtheorem{cor}{Corollary}
\theoremstyle{definition}
\newtheorem{defn}{Definition}[section]
\newtheorem{exmp}{Example}[section]
\newtheorem{numexmp}{Numerical Example}[section]

\def\sumn{\sum_{i=1}^n}

\title{Convergence Rate Bounds for Iterative Random Functions Using One-Shot Coupling}
\author{Sabrina Sixta and Jeffrey S. Rosenthal}
\def\mc{Markov chain }

\def\wrt{with respect to }

\def\vecx{\vec{X}}
\def\L{\mathcal{L}}
\def\X{\mathcal{X}}

\newcommand{\indep}{\perp \!\!\! \perp}
\def\tin{\tau^{-1}}

\def\ax{p/2}
\def\bx{C/2}
\def\ay{(k+p)/2}
\def\by{C/2}
\def\tinv{\tau^{-1}}
\def\aax{1/2}
\def\bbx{S/2}
\def\aay{(J+2)/2}
\def\bby{S/2}
\def\borel{\mathcal{B}}

\DeclarePairedDelimiter\abs{\lvert}{\rvert}
\DeclarePairedDelimiter\norm{\lVert}{\rVert}
\DeclarePairedDelimiter\Abs{\bigg\lvert}{\big\rvert}

\begin{document}
	\maketitle

\begin{abstract}
	One-shot coupling is a method of bounding the convergence rate between two copies of a Markov chain in total variation distance, which was first introduced in \cite{oneshot} and generalized in \cite{wasandoneshot}. The method is divided into two parts: the contraction phase, when the chains converge in expected distance and the coalescing phase, which occurs at the last iteration, when there is an attempt to couple. One-shot coupling does not require the use of any exogenous variables like a drift function or a minorization constant. In this paper, we summarize the one-shot coupling method into the One-Shot Coupling Theorem. We then apply the  theorem to two families of Markov chains: the random functional autoregressive process and the autoregressive conditional heteroscedastic (ARCH) process. We provide multiple examples of how the theorem can be used on various models including ones in high dimensions. These examples illustrate how the theorem's conditions can be verified in a straightforward way. The one-shot coupling method appears to generate tight geometric convergence rate bounds.
\end{abstract}

\tableofcontents

\section{Introduction}

The study of Markov chain convergence rates focuses on evaluating how fast a positive recurrent Markov chain converges to its stationary distribution. On one hand, a great deal of progress has been made in bounding the convergence rate for Markov chains defined in discrete state spaces \cite{discretemc, eigen, firstlook}. On the other hand, despite the major developments made in bounding Markov chains in continuous state space, many applications of continuous state space Markov chains do not have established convergence rate bounds. For example, convergence rate bounds applied to Markov chain Monte Carlo (MCMC) models, which are useful for deciding the size of the burn-in period \cite{honestexp, intromcmc}, do not have known upper bounds on their convergence rate \cite{intromcmc}. Users need to rely on ad-hoc convergence diagnostics (e.g., \cite{gelmanrubin}), which offer no guarantees.


Methods using the drift and minorization conditions (e.g., \cite{rose, baxendale}), which guarantee geometric ergodicity (definition \ref{def:gemoerg}), are the most studied techniques for bounding Markov chains in continuous state space \cite{robrose, honestexp}. The minorization condition is satisfied for a Markov chain $\{X_n\}_{n\geq 1}$ under the following circumstances: there exists a small set $K$, a probability measure $Q$ and a positive number $\epsilon>0$ such that $P(\cdot\mid X_n=x)\geq \epsilon Q(\cdot)$ for $x\in K$. The drift condition is satisfied if there exists a positive function $V$, and constants $\alpha>1$ such that $E[V(X_{n+1})\mid X_{n}=x]\leq V(x)/\alpha$ \cite{meyneandtweedie, robrose}. Bounds generated using the drift and minorization conditions have been applied to a wide array of problems such as \cite{stein, geomgibbs, honestexp}.

Despite the widespread use of bounds generated by the drift and minorization conditions, there are drawbacks. First, it can be a challenge to identify a small set $K$ and drift function $V$ \cite{wasandoneshot}. Second, it is shown in \cite{wasmeth} based on results from \cite{jerison} that bounds that use the minorization condition do not scale well in high dimensions. 

Alternatively, methods for finding Markov chain convergence rate bounds on the Wasserstein distance have been shown to scale well in high dimensions \cite{durmuswasconv, wasmeth}, so bounding the total variation distance by first bounding the Wasserstein distance is a common technique used in the literature \cite{wasmeth, wasandoneshot,waslinearmixedmodels}. 

One-shot coupling, which bounds the Wasserstein distance as an intermediate step \cite{wasandoneshot}, provides an upper bound on the convergence rate in total variation distance of a Markov chain. This method does not need to identify any exogenous sets or functions, like the drift and minorization conditions. Further, the one-shot coupling method has already been shown to scale well in certain high dimensional examples \cite{oneshot, oneshotnsphereexample} and will be shown in this paper to scale well in high dimensions for the Bayesian regression Gibbs sampler (Example \ref{ex:bayesianreggibbs}) and the Bayesian location Gibbs sampler (Example \ref{ex:anothergibbssampler}).

The one-shot coupling method described in \cite{oneshot} works by first converging the expected distance between two copies of a Markov chain. At the last iteration, the probability of coupling is evaluated when the expected distance between the copies is small. This contrasts with the drift and minorization technique, which attempts to couple the two Markov chain copies every time they enter some fixed small set $K$. 


In this paper, we introduce the One-Shot Coupling Theorem in Section \ref{subsec:oneshotcoupling}, which aims to summarize the method defined in \cite{oneshot} and \cite{wasandoneshot} for straightforward applications. The One-Shot Coupling Theorem is used as the foundation for bounding the convergence rate for all of the examples in this paper, which can be partitioned into two families: the random functional autoregressive process and the ARCH process. In Section \ref{sec:rfautoregproc} we introduce the Sideways Theorem \ref{thm:oneshotlipschitz}, which is new and is an application of the One-Shot Coupling Theorem. We apply it to various examples of random functional autoregressive processes (definition \ref{def:ranfunARP}). In Section \ref{sec:linproc} we provide convergence rate bounds using the One-Shot Coupling Theorem to various ARCH processes (definition \ref{def:ARCH}).

Proofs for the theorems presented in this paper are found in the appendix, Section \ref{sec:appendix}. The code used to generate all of the tables and calculations can be found on \href{https://github.com/sixter/OneShotCoupling}{github.com/sixter/OneShotCoupling}.
\section{Background and notation} \label{subsec:notation}

Let  $\{X_n\}_{n\geq 1}$ and $\{X'_n\}_{n\geq 1}$ be two copies of the Markov chain over the state space $\X$ and define $\L(X_n)$ to be the distribution of the random variable $X_n$. We define $\pi$ to be the stationary distribution of the Markov chain.
%
\subsection{Total variation distance}\label{sub:tv}
We are interested in measuring the distance between the distribution of two Markov chains. To measure this we use the total variation metric.

\begin{defn}[Total variation distance]
	The total variation distance between the laws of two random variables, $X$ and $X'$, defined on the state space $\X$ is 
	$$\norm{\L(X)-\L(X')} = \sup_{A \subseteq \X}\abs{P(X \in A)-P(X' \in A)}$$
	where $\L(X)$ represents the distribution of the random variable $X$ and $A$ is a measurable set.
\end{defn}
For random variables, $X,X'\in \mathbb{R}$ with defined density functions $f_{X}, f_{X'}$ and reference measure $\lambda$, 
\begin{equation}\label{eqn:tv}
	\norm{\L(X)-\L(X')} = \frac{1}{2}\int_{\mathbb{R}}\abs{f_{X}(x)-f_{X'}(x)}\lambda (dx)
\end{equation}

Total variation distance has natural probability interpretations. It is the maximum difference in probabilities of an event. It is the error in an expected bounded loss function when a given measure is used as a proxy for another \cite{gibbs}. Finally, it can be seen as the percentage of samples of $\L(X)$ which cannot be regarded as samples from $\L(X')$ (Proposition 3(g) \cite{robrose}). 

Historically, total variation distance was the common metric for measuring Markov chain convergence rates \cite{robrose, meyneandtweedie, clt, honestexp} and hence, there is a rich literature of attributes that can be deduced from finding convergence rates in total variation. For example, mixing times in total variation distance can be used to determine whether the Markov chain is asymptotically uncorrelated (Theorem 2 of \cite{clt}), uniformly integrable (Theorem 3 of \cite{clt}), whether the central limit theorem (CLT) applies (Theorem 9 of \cite{clt} or Section 5.2 of \cite{robrose}), or whether it is convergent based on the total variation mixing times of another Markov chain (Theorem 8 of \cite{comp}). 

The following are properties of total variation, which will be used in conjunction with the One-Shot Coupling Theorem \ref{thm:oneshot} to establish upper bounds on the convergence rate for the examples in this paper. 

Proposition \ref{prop:invertibletv} states that the total variation between two random variables is equal to the total variation of any invertible transform of the same random variables. This proposition resembles Lemma 4.13 of \cite{markovmixing} and Lemma 3 of \cite{wasandoneshot}.

\begin{prop}\label{prop:invertibletv}
	Let $X,X'\in \mathcal{X}$ be two random variables and let $g: \mathcal{X}\to \mathcal{Y}$ be an invertible and measurable function. 
	Then,
	\begin{equation}\label{eqn:invertibletv}
		\norm{\L(g(X))-\L(g(X'))}=\norm{\L(X)-\L(X')}
	\end{equation}
The proof is in Section \ref{proof:invertibletv}.
\end{prop}

In general, for a measurable (not necessarily invertible) function $g$, $g^{-1}(f(\borel))\subset \borel$, so the third equality in the proof becomes $\leq$ and $$\norm{\L(g(X))-\L(g(X'))} \leq \norm{\L(X)-\L(X')}$$

Proposition \ref{prop:tvexp} states that the total variation distance between two random variables is bounded above by the expected value of the conditional random variable.

\begin{prop}\label{prop:tvexp}
	Let $X,X'$ be two random variables with corresponding $\sigma$-field $\borel$ and $Y\in \mathcal{Y}$ be some related random variable. Then
	\begin{equation*}
		\norm{\L(X)-\L(X')}\leq E\left[ \norm{\L(X\mid Y)-\L(X'\mid Y)} \right]
	\end{equation*}
The proof is in Section \ref{proof:tvexp}.
\end{prop}

Proposition \ref{prop:indcoord} states that the convergence rate of a Markov chain in $\mathbb{R}^d$ with independent coordinates is $d$ times the maximum coordinate-wise convergence rate.
This proposition is an application of inequality 1.2 of \cite{prodmeasures}.

\begin{prop}\label{prop:indcoord}
	Let $\{\vecx_n\}_{n\geq 1}\in \mathbb{R}^d$ be a Markov chain such that each coordinate is independent of the other coordinates, $X_{i,n} \indep X_{j,n}, i\neq j$. Further suppose that for two copies of the Markov chain $\{\vecx_n\}_{n\geq 1}$ and $\{\vecx'_n\}_{n\geq 1}$,
	$\max_{1\leq i\leq d}\norm{\L(X_{i,n})-\L(X'_{i,n})}\leq Ar^n$
	for some $A\in \mathbb{R}_+$ and $r\in(0,1)$. Then, 
	\begin{equation}\label{eqn:indcoord}
		\norm{\L(\vec{X}_n)-\L(\vec{X}'_n)}\leq dAr^n
	\end{equation}
	The proof is in Section \ref{proof:indcoord}.
\end{prop}

In this paper, we establish convergence bounds for Markov chains that are geometrically ergodic in total variation distance. 
\begin{defn}[Geometric ergodicity]\label{def:gemoerg}
	Let $\{X_n\}_{n\geq 1}$ be a Markov chain with a stationary distribution $\pi$. The Markov chain is geometrically ergodic if there exists a $\rho<1$ and a function $M(x)<\infty$, $\pi$-a.e. such that for $X_0=x$,
	\begin{equation}\label{eqn:geomerg}
		\norm{\L(X_n)-\pi}\leq M(x)\rho^n
	\end{equation}
\end{defn}
The geometric rate of convergence for $X_n$ is defined as $\rho^* = \inf\{\rho : \text{ equation \ref{eqn:geomerg} holds}\}$.

Proposition 4 of \cite{wasmeth} states that for any sequence of drift and minorization conditions, the geometric convergence rate $\rho$ established by the Rosenthal bound (Theorem 12 of \cite{rose}) will increase at an exponential rate for the autoregressive normal process in $\mathbb{R}^d$ as the dimension $d\to \infty$. This finding suggests that convergence bounds that use the drift and minorization condition do not scale well in dimension (see Lemma 3 and discussion in \cite{wasmeth}). However, Proposition \ref{prop:indcoord} shows that since each coordinate in this example is independent, the geometric convergence $\rho$ rate is indeed invariant to dimension, regardless of the bounding approach. Thus a drift and minorization bound, including the Rosenthal bound, can easily be applied to the autoregressive normal process in $\mathbb{R}$ and then extended to $\mathbb{R}^d$ using Proposition \ref{prop:indcoord}. To see Proposition \ref{prop:indcoord} applied to the autoregressive normal process in $\mathbb{R}^d$, see Example \ref{ex:arnormind}.

\subsection{Wasserstein distance}
Let $X,X'\in \mathbb{R}$ be two random variables equipped with the Euclidean distance. The Wasserstein distance is defined as follows, 
$$W(\L(X), \L(X')) = \inf\{E[\abs{Y-Y'}] : \L(X)=\L(Y)\text{ and }\L(X')=\L(Y')\}$$

In comparison to total variation distance, there is not as much literature dedicated to Markov properties that can be derived from the convergence in Wasserstein distance. However, this literature \textit{is} growing. For example Jin and Tan provide sufficient conditions in  \cite{Jin2020CentralLT} for the CLT based on convergence in Wasserstein distance (see Theorems 9 and 10). 

%

\subsection{Coupling} \label{subsec:coupling}

Total variation can also be defined in terms of the coupling characterization \cite{gibbs},
\begin{align*}
	\norm{\L(X)-\L(X')} = \inf \{P(Y\neq Y') \mid \L(X)=\L(Y) \text{ and } \L(X')\neq \L(Y')\}
\end{align*}

%

The total variation metric measures the distance between two distributions, but is invariant to \emph{how} these measures are jointly distributed. For example, let $X\sim N(0,1)$ and $X' \sim N(1,1)$ be two random variables. Regardless of whether $X$ and $X'$ are highly dependent, for example if $X=X'+1$ or if $X, X'$ are independent, their total variation distance would be the same. The Nummelin splitting technique makes use of this by constructing alternative random variables, $Y$ and $Y'$, such that the marginal distributions are the same $\L(X)=\L(Y)$, $\L(X')=\L(Y')$, and the probability that they are unequal is minimised. This technique was first shown in \cite{numsplitting}. See \cite{eigen} or \cite{meyneandtweedie} for an explanation. Finally, note that the theory on maximal coupling guarantees that there exists alternative random variables $Y,Y'$ as defined above, such that $	\norm{\L(X)-\L(X')} = P(Y\neq Y')$ \cite{maximalcoupling}.

Coupling techniques are widely used to calculate total variation upper bounds on Markov chains \cite{firstlook, rose,eigen,robrose, stein, jun}. Let $\{X_n\}_{n\geq 1}$ and $\{X'_n\}_{n\geq 1}$ be two copies of a Markov chain. If we want to use the coupling characterization for finding an upper bound on the total variation distance, we must also make sure that the faithful coupling condition holds (see Section 2 of \cite{faith}).
That is, for any measurable set $A \in \mathcal{X}$,
\begin{align*}
	P(X_{n+1} \in A : X_n = x\text{ and } X'_n=x') &= P(X_{n+1} \in A : X_n = x)\\
	P(X'_{n+1} \in A : X_n = x\text{ and } X'_n=x') &= P(X'_{n+1} \in A : X'_n = x')
\end{align*}
If the faithful coupling condition holds, then calculating the probability that two chains are unequal at iteration $n$ can be interpreted as the probability that the two chains have not yet coupled by iteration $n$. This is because once the two Markov chains couple, they can be structured so that they are equal forever and so $P(X_n \neq X'_n)=P(T\geq n)$ where $T=\min\{k : X_k=X'_k\}$ (Theorem 1 of \cite{faith}). If a minorization condition holds on the Markov chain, then the faithful coupling condition also holds. For one-shot coupling, we do not need faithful coupling, because we only try to couple the chains at the last iteration.

\section{One-Shot Coupling} \label{subsec:oneshotcoupling}


	One-shot coupling is an alternative way of applying coupling methods to bound the total variation of two copies of a Markov chain. To apply one-shot coupling, we define a Markov chain  in terms of iterated random functions \cite{iteratedranfunc}. That is, define a family of random functions $\{g_{\vec{\theta}} : \vec{\theta} \in \pmb{\Theta}\}$ such that $\vec{\theta}_n=\{\theta_{1,n},\ldots, \theta_{d,n}\}$ is a random vector and 
$$X_{n}=g_{\vec{\theta}_{n}}(X_{n-1})$$
The $n$th iteration of the Markov chain  can be written in terms of $X_0=x$ as follows,
$$X_n = (g_{\vec{\theta}_{n}}\circ g_{\vec{\theta}_{n-1}}\ldots \circ g_{\vec{\theta}_{1}})(x) = g_{\vec{\theta}_{n}}(g_{\vec{\theta}_{n-1}}(\ldots g_{\vec{\theta}_{1}}(x) \ldots))$$

Summarizing Section 3 of \cite{oneshot}, to find an upper bound on the total variation distance between $X_N$ and $X'_N=g_{\vec{\theta}'_{N}}(X'_{N-1})$ we do the following.
\begin{enumerate}
	\item \textbf{Contraction phase: }For $n<N$, set $\vec{\theta}_{n}=\vec{\theta}'_{n}$ so that the two chains get `closer' together. 
	\item \textbf{Coalescing phase: }For $n=N$, we set all but one coordinate of $\vec{\theta}_{n}$ and $\vec{\theta}'_{n}$ to equality and attempt to couple $X_n$ and $X'_n$. That is, specify coordinate $j\in\{1,\ldots, d\}$ and set $\theta_{i,n}=\theta'_{i,n}$ for all $i\neq j$. Then we attempt to jointly choose $\theta_{j,n}$ and $\theta'_{j,n}$, such that
	$$g_{(\theta_{1,n},\ldots,\theta_{j,n}, \ldots, \theta_{d,n})}(X_{n-1}) = g_{(\theta_{1,n},\ldots,\theta'_{j,n}, \ldots, \theta_{d,n})}(X'_{n-1})$$
\end{enumerate} 

The method used in the contraction phase is also known as the common random number method and is discussed in detail in Section 2.3.1 of \cite{jacob}. The contraction phase can also be used to directly generate upper bounds in Wasserstein distance \cite{jacob, geomconvrates, gibbswass} (it is also used to generate bounds on other types of distances like Monge–Kantorovich or Prokhorov \cite{jacob}). 

The one-shot coupling method has been applied over a variety of specific examples, namely, a nested gamma model in \cite{oneshotnestingexmple}, an image restoration model in \cite{oneshotimageexample}, and a random walk on the unit sphere in \cite{oneshotnsphereexample}.

The contraction and coalescing phase described above is how the one-shot coupling method was first defined in \cite{oneshot}. The following theorem summarizes the above method and serves as a general outline for bounding the total variation distance between two Markov chains. The coalescing condition below does not specify \textit{how} the two chains will couple, unlike the method described above.

\begin{thm}[One-Shot Coupling Theorem]\label{thm:oneshot}
	Let $\{X_n\}_{n\geq 1}, \{X'_n\}_{n\geq 1}$ be two copies of a Markov chain such that $X_{n}=g_{\theta_{n}}(X_{n-1})$ and $X'_{n}=g_{\theta'_{n}}(X'_{n-1})$, where $(\theta_n, \theta'_n)_{n\geq 1}$ are independent random variables \wrt $n$ and the marginal distribution of $\theta_n, \theta'_n \sim \mathcal{D}$, for some distribution $\mathcal{D}$. Suppose that the following two conditions hold for some non-negative integer $n_0$.
	
	\begin{enumerate}
		\item \textbf{Contraction condition:} There exists a $D\in (0,1)$ such that for any $n\geq n_0$ when $\theta_{n+1}=\theta'_{n+1}\sim \mathcal{D}$
		$$E[\abs{g_{\theta_{n+1}}(X_{n})-g_{\theta_{n+1}}(X'_{n})}]\leq DE[\abs{X_{n}-X'_{n}}]$$

		\item \textbf{Coalescing condition:} There exists a $C>0$ such that for any $n\geq n_0$
		$$\norm{\L(X_{n})-\L(X'_{n}))}\leq CE[\abs{X_{n-1}-X'_{n-1}}]$$
	\end{enumerate}
Then the total variation distance between the two Markov chains at iteration $n\geq n_0$ is 
$$ \norm{\L(X_{n})-\L(X'_{n})}\leq C D^{n-n_0-1} E[\abs{X_{n_0}-X'_{n_0}}] $$
\end{thm}

\begin{proof}[Proof of the One-Shot Coupling Theorem \ref{thm:oneshot}]
		Fix $n\geq n_0$. We are interested in finding an upper bound on $\norm{\L(X_{n})-\L(X'_{n})}$. To do so, we first generate alternative random variables, $Y_n, Y'_n$ such that
		\begin{enumerate}
				\item for $0\leq m \leq n_0$: $Y_{m}=X_{m}, Y'_{m}=X'_{m}$ 
				\item for $n_0 < m< n$: $\theta_{m}=\theta'_{m}\sim \mathcal{D}$ and $Y_{m}=g_{\theta_{m}}(Y_{m-1}), Y'_{m}=g_{\theta_{m}}(Y'_{m-1})$.
				\item for $m=n$:  $\theta_{m},\theta'_{m}\sim \mathcal{D}$ with an arbitrary joint distribution and $Y_{m}=g_{\theta_{m}}(Y_{m-1}), Y'_{m}=g_{\theta'_{m}}(Y'_{m-1})$
			\end{enumerate}
		By construction, $Y_m\overset{d}{=}X_m$ and $Y'_m\overset{d}{=}X'_m$ for $0\leq m\leq n$.
		
		Next we find an upper bound on the total variation distance between $Y_{n}$ and $ Y'_{n}$. 
		By the contraction condition for $n_0\leq m <n$, $E[\abs{g_{\theta_{m+1}}(Y_{m})-g_{\theta_{m+1}}(Y'_{m})}]\leq DE[\abs{Y_{m}-Y'_{m}}]$ and so,
		\begin{align*}
				E[\abs{Y_{n-1}-Y'_{n-1}}] = E[\abs{g_{\theta_{n-1}}(Y_{n-2})-g_{\theta_{n-1}}(Y'_{n-2})}] \leq D E[\abs{Y_{n-2}-Y'_{n-2}}]\leq D^{n-n_0-1} E[\abs{Y_{n_0}-Y'_{n_0}}]
			\end{align*}
		
		By the coalescing condition,
		\begin{align*}
				\norm{\L(Y_{n})-\L(Y'_{n}))}&\leq CE[\abs{Y_{n-1}-Y'_{n-1}}]\leq C D^{n-n_0-1} E[\abs{Y_{n_0}-Y'_{n_0}}]=C D^{n-n_0-1} E[\abs{X_{n_0}-X'_{n_0}}]
			\end{align*}
		Finally since $Y_n\overset{d}{=}X_n$ and $Y'_n\overset{d}{=}X'_n$,
		$$\norm{\L(X_{n})-\L(X'_{n})}=\norm{\L(Y_{n})-\L(Y'_{n})}\leq C D^{n-n_0-1} E[\abs{X_{n_0}-X'_{n_0}}]$$
\end{proof}

If $\L(X_n)$ has a density function with respect to $X_{n-1}=x$, $f(x,z)$, then Theorem \ref{thm:oneshot} can be proven with Wasserstein distance as an intermediary using the following lemma. 

\begin{lem}[Theorem 12 of \cite{wasandoneshot}]\label{lem:coalescwas}
	If $\frac{1}{2}\int_{\mathcal{X}} \abs{f(x,z)-f(x',z)}\lambda(dx)\leq C \abs{x-x'}$ holds, then for $n\geq 0$
	$$\norm{\L(X_n)-\L(X'_n)}\leq CW(\L(X_{n-1}), \L(X'_{n-1}))$$
\end{lem}

If the contraction condition holds, then for $n\geq n_0$, $W(\L(X_{n-1}), \L(X'_{n-1})) \leq E[\abs{X_{n-1}-X'_{n-1}}] \leq D^{n-n_0-1}E[{X_{n_0}-X'_{n_0}}]$ and the proof of Theorem \ref{thm:oneshot} directly follows.

In most cases $n_0=0$. See the GARCH Example \ref{ex:garch} for an alternative case, $n_0=1$.

In general, the contraction condition can be weakened. Theorem 1.1 of \cite{iteratedranfunc} provides sufficient conditions to guarantee the existence of $D$ as defined in the above theorem. The conditions in Theorem 1 of \cite{loccont}, which are called local contractivity and are weaker, could also replace the contraction condition in the above theorem.

To bound the total variation between a Markov chain, $\{X_n\}_{n\geq 1}$, and the corresponding stationary distribution, $\pi$, we set $X'_0 \sim \pi$. This implies that $X'_n\sim \pi$ and 
$\norm{\L(X_n)-\pi}\leq C D^{n-n_0-1} E_{X_{\infty} \sim \pi}[\abs{X_{n_0}-X_{\infty}}]$ where $C, D$, and $n_0$ are satisfied according to the conditions above. 

To find an upper bound on $E_{X_{\infty} \sim \pi}[\abs{X_{n_0}-X_{\infty}}]$ we use the following Lemma \ref{lem:locmodelexpdist}, which uses a drift condition to bound the expected distance between the stationary distribution of a Markov chain and an initial value.

\begin{defn}[Drift condition]\label{def:drift}
	Let $\{X_n\}_{n\geq 1}$ be a Markov chain on $\mathcal{X}$. A drift condition is satisfied if there exists a function $V:\mathcal{X}\to \mathbb{R}$ and constants $\lambda\in (0,1)$ and $b<\infty$ such that 
	$E[V(X_n)\mid X_{n-1}]\leq \lambda V(X_{n-1}) +b $.
\end{defn}

\begin{lem}\label{lem:locmodelexpdist}
	Let $\{X_n\}_{n\geq 1}$ be a Markov chain such that a drift condition \ref{def:drift} holds with $V(x)=(x+h)^2, h\in \mathbb{R}$. The expected distance between $X_0$ and $X_{\infty}\sim \pi$ is bounded above as follows,
	$E[\abs{X_{\infty}-X_0}]\leq \sqrt{\frac{b}{1-\lambda}}+ E[\abs{X_0+h}]$.
	\begin{proof}
		$	E[\abs{X_{\infty}-X_0}] \leq E[\abs{X_{\infty}+h}]+E[\abs{X_0+h}]\leq \sqrt{\frac{b}{1-\lambda}}+E[\abs{X_0+h}]$. The last inequality holds by lemma \ref{lem:expstat}.
	\end{proof}
\end{lem}

\begin{lem}\label{lem:expstat}[Proposition 4.3 (i) of \cite{meyneandtweedie}]
	If the drift condition holds, then $E_{\pi}[V(X)]\leq \frac{b}{1-\lambda}$.
	See Section \ref{proof:expstat} for a proof.
\end{lem}

See Numerical Example \ref{numex:locmodel} for an application of Lemma \ref{lem:locmodelexpdist}.


\section{Random-functional autoregressive processes}\label{sec:rfautoregproc}
The following section proposes the Sideways Theorem to generate 
upper bounds on the total variation distance for random-functional autoregressive processes.

\begin{defn}[Random functional autoregressive processes]\label{def:ranfunARP}
	The sequence $\{X_n\}_{n\geq 1}$ is a random functional autoregressive process if for $g:\mathbb{R}^2\to \mathbb{R}$
	\begin{equation*}
		X_{n}=g(\theta_{1,n},X_{n-1})+\theta_{2,n}
	\end{equation*}
	where $(\theta_{1,n},\theta_{2,n})\in \mathbb{R}^2$ are random variables and $(\theta_{1,n},\theta_{2,n})\indep (\theta_{1,m},\theta_{2,m})$ when $n\neq m$.
\end{defn}

\begin{thm}[Sideways Theorem]\label{thm:oneshotlipschitz}
	Let $\{X_n\}_{n\geq 1}\in \mathbb{R}$ be a random-functional autoregressive.
	Suppose that,
	\begin{enumerate}
		\item \textbf{Contraction condition:} There exists a $D\in (0,1)$ such that for $n\geq 0$, $$E[\abs{g(\theta_{1,n+1},X_{n})-g(\theta_{1,n+1},X'_{n})}]\leq DE[\abs{X_n - X'_n}]$$
		\item \textbf{Attributes of the conditional density $\theta_{2,n}\mid \theta_{1,n}$:} The conditional density of $\theta_{2,n}\mid \theta_{1,n}$
		\begin{enumerate}
			\item is bounded above: There exists a $K>0$ such that for all $(\theta_{1,n},\theta_{2,n})\in \mathbb{R}^2$, the conditional density function of $\theta_{2,n}$ is bounded above by $K$, $f_{\theta_{2,n}}(\theta_{2,n}\mid \theta_{1,n})\leq K$.
			\item has at most $M$ local extrema points that are at most $L>0$ distance apart: For any $\theta_{1,n}$, there are $M$ local maximas and minimas (local extrema points) within the conditional density. The local extrema points are at most $L$ distance apart.
			\item is continuous for any $\theta_{1,n}$
		\end{enumerate}
	\end{enumerate}
	Then an upper bound on the geometric rate of convergence of the \mc is $D$ and the total variation distance between the two copies of the Markov chain, $X_n, X'_n$, is bounded above as follows,
	\begin{equation}\label{eqn:oneshotlipschitz}
		\norm{\L(X_{n})-\L(X'_{n})}\leq \left(\frac{K(M+1)}{2} +\frac{I_{M>1}}{L}\right)D^{n-1} E[\abs{X_0-X'_0}]
	\end{equation}
\end{thm}

The attributes of the conditional density of $\theta_{2,n}\mid \theta_{1,n}$ serve to prove, by integrating along the $y$-axis or flipping the density sideways, that the coalescing condition is satisfied. To prove the Sideways Theorem, we show that the contraction and coalescing conditions are satisfied and then apply the One-Shot Coupling Theorem \ref{thm:oneshot}.

\begin{lem}[Coalescing condition]\label{lem:coalescingcondsideways}
	If the density of $\theta_{2,n}|\theta_{1,n}$ for any $\theta_{1,n}$ is (1) bounded above, (2) has at most $M$ local extrema points that are at most $L$ distance apart, and (3) is continuous then for $n\geq 0$,
	$$\norm{\L(X_{n})-\L(X'_{n})}\leq C E[\abs{X_{n-1}-X'_{n-1}}]$$
	Where $C=\frac{K(M+1)}{2} +\frac{I_{M>1}}{L}$.
	See Section \ref{pf:coalescingcondsideways} for a proof.
\end{lem}

\begin{proof}[Proof of Theorem \ref{thm:oneshotlipschitz}]
	The following shows that the contraction condition holds for $D\in (0,1)$ and $n\geq 0$,
	\begin{align*}
		E[\abs{f_{\theta_{n}}(X_{n-1})-f_{\theta_{n}}(X'_{n-1})}]&=E[\abs{(g(\theta_{1,n},X_{n-1})+\theta_{2,n})-(g(\theta_{1,n},X'_{n-1}) +\theta_{2,n})}]\\ 
		&=E[\abs{g(\theta_{1,n},X_{n-1})-g(\theta_{1,n},X'_{n-1})}]\\
		&\leq DE[\abs{X_{n-1} - X'_{n-1}}] &\text{by contraction condition}
	\end{align*}
	Lemma \ref{lem:coalescingcondsideways}, which can be applied when condition 2 is satisfied (attributes of the conditional density of  $\theta_{2,n}\mid \theta_{1,n}$), shows that the coalescing condition holds. By the One-Shot Coupling Theorem \ref{thm:oneshot}, the total variation distance between two copies of the process can be bounded above using equation \ref{eqn:oneshotlipschitz}.
\end{proof}
%



In \cite{arconvrate}, it is shown that when the function $g$ is deterministic ($g$ is a function of $X_{n-1}$ only and not $\theta_{1,n}$) and given the same assumptions on $\theta_{2,n}$, the upper bound on the  geometric rate of convergence is $D$ (see Corollary 8 and Example 9 of \cite{arconvrate}). This matches the results from our theorem.

Note that the Sideways Theorem \ref{thm:oneshotlipschitz} provides an upper bound on total variation distance, but does not imply the existence of a stationary distribution for the Markov chain. To develop the intuition for this, first note that convergence in total variation distance implies convergence in distribution \cite{gibbs}. Suppose that $\L(X_n), \L(X'_n)$ have distribution functions, $F_n, F'_n$, then by Helly's Selection Theorem (see Lemma 11.1.8 of \cite{firstlook}), a right continuous function $F$ exists such that $F_n\to F$ and $F'_n\to F$ pointwise. However, the function $F$ may not necessarily be a distribution function. This is an illustration of why a stationary distribution may not exist. 

A simple counter example would be the process $X_n=\frac{1}{2}X_{n-1}+n+Z_n, Z_n\sim N(0,1)$ where $g(\theta_{1,n},X_n)=\frac{1}{2}X_{n-1}+n$ and $\theta_{2,n}=Z_n$. It is clear how the Sideways Theorem \ref{thm:oneshotlipschitz} could generate a geometric convergence bound over two iterations of the process if $E[X_0-X'_0]<\infty$, but $X_n,X'_n\to \infty$ almost surely and so there is no stationary distribution. See \cite{loccont} for more information on sufficient conditions for stationarity.

%
%

%


\subsection{An example of a non-linear autoregressive process}
\begin{exmp}[Non-linear autoregressive process]
	This example is discussed in Section 4 of \cite{geomconvrates}. Let $\{X_n\}_{n\geq 1}$ be a Markov chain such that 
	$$X_{n+1}=\frac{1}{2}(X_n -\sin X_n)+Z_{n+1}$$
	where $\{Z_n\}_{n\geq 1}\sim N(0,1)$ are independent and identically distributed (i.i.d.) random variables. In \cite{geomconvrates}, it is assumed that $\{Z_n\}_{n\geq 1}$ are i.i.d. random variables with a mean of 0 and a variance of 1.
	
	For $g(x)=\frac{1}{2}(x -\sin(x))$, the derivative is $g'(x)=\frac{1}{2}(1 -\cos(x))$ and so $\sup_{x \in \mathbb{R}}g'(x)=1$. This cannot be used. Instead, we can find a value for $D$ in terms of the second iteration. That is,
	\begin{equation*}
		D^2 = \sup_{x,y} \frac{E[\abs{X_{n+2}-X'_{n+2}} \mid X_n=x,X'_n=y]}{\abs{x-y}}
	\end{equation*}

\begin{lem}\label{lem:nonlinar}
	The value of $D$ as defined above can be written as $$D^2=\sup_{x,y} \frac{\sqrt{4h(x,y)^2 - 8e^{-1/2}h(x,y) \sin h(x,y)\cos k(x,y) + 2\sin^2 h(x,y) (1+e^{-2}(\cos^2k(x,y) - \sin^2 k(x,y)))}}{2\abs{x-y}}$$
	where
	\begin{equation*}
		h(x,y)=\frac{1}{4}(y-x+\sin x - \sin y) \hspace{1cm} k(x,y)=\frac{1}{4}(x+y-\sin y - \sin x)
	\end{equation*}
The proof can be found in Section \ref{proof:nonlinar}.
\end{lem}

Using simulation, we can deduce that $D^2\approx 0.813^2=0.661$, which closely matches the geometric convergence rate found in \cite{geomconvrates} for the Wasserstein distance of $D=0.814$.

Using the Sideways Theorem \ref{thm:oneshotlipschitz} notation, $K=\frac{1}{\sqrt{2\pi}}$ and $M=1$. An upper bound on the total variation distance is 
\begin{equation*}
	\norm{\L(X_{n+1})-\L(X'_{n+1})}\leq \frac{1}{\sqrt{2\pi}} E[\abs{X_0-X'_0}] 0.661^{\left \lfloor{N/2}\right \rfloor }
\end{equation*}

Thus if $X_0=1$ and $X'_0=2$, then after 20 iterations, the total variation distance between the two processes will be less than 0.01.

\end{exmp}

\subsection{Random-coefficient autoregressive models}
\begin{cor}\label{cor:oneshotlipschitz}
	Let $\{X_n\}_{n\geq 1}\in \mathbb{R}$ be a random-coefficient autoregressive model. That is, $X_n$ is of the following form
	$$X_{n}=\theta_{1,n}X_{n-1}+\theta_{2,n}$$
	where $(\theta_{1,n},\theta_{2,n})\indep (\theta_{1,m},\theta_{2,m})$ when $n\neq m$. If we replace the contraction condition of the Sideways Theorem \ref{thm:oneshotlipschitz} with  
	\begin{enumerate}
		\item $E[\abs{\theta_{1,n}}]<1$
	\end{enumerate}
	Then equation \ref{eqn:oneshotlipschitz} holds for $D=E[\abs{\theta_{1,n}}]$.
\end{cor}
\begin{proof}
	If $E[\abs{\theta_{1,n}}]<1$ then set $D=E[\abs{\theta_{1,n}}]$ and so the contraction condition in Theorem \ref{thm:oneshotlipschitz} holds,
	$$E[\abs{g(\theta_{1,n+1},X_{n})-g(\theta_{1,n+1},X'_{n})}]=E[\abs{\theta_{1,n+1}X_{n}-\theta_{1,n+1}X'_{n}}]\leq D E[\abs{X_n - X'_n}]$$
	Since all of the conditions in Theorem \ref{thm:oneshotlipschitz} are satisfied, equation \ref{eqn:oneshotlipschitz} holds. 
\end{proof}


\subsection{Bayesian regression Gibbs sampler}
\begin{exmp}[Bayesian regression Gibbs sampler]\label{ex:bayesianreggibbs}
	Suppose we have the following observed data $Y\in \mathbb{R}^k$ and $X\in \mathbb{R}^{k\times p}$ where
	$$Y \mid \beta, \sigma^2 \sim N_{k}(X\beta, \sigma^2I_k)$$
	for unknown parameters $\beta\in \mathbb{R}^p, \sigma^2\in \mathbb{R}$. Suppose we apply the prior distributions on the unknown parameters,
	\begin{itemize}
		\item $\beta \mid \sigma^2 \sim N_p(0_p, \frac{\sigma^2}{\lambda}I_p)$, where $\lambda>0$ is known 
		\item $\pi(\sigma^2)\propto 1/\sigma^2$
	\end{itemize}
The Bayesian regression Gibbs sampler is based on the conditional posterior distributions of $\beta_n, \sigma^2_n$ and is defined as follows.
	\begin{itemize}
		\item $\beta_n\mid \sigma_{n-1}^2,Y \sim N_p(\tilde{\beta}, \sigma^2_{n-1} A^{-1})$
		\item $\sigma^2_n\mid \beta_n,Y \sim \Gamma^{-1}\left(\frac{k+p}{2},\frac{1}{2}\left[(\beta_n-\tilde{\beta})^T A(\beta_n-\tilde{\beta})+C\right]\right)$. $\Gamma^{-1}(\alpha, \beta)$ represents the inverse gamma distribution.
	\end{itemize}
	Where $A=X^TX+\lambda I_p$ is positive semi-definite, $\tilde{\beta}=A^{-1}X^TY$, and $C=Y^T(I_k-XA^{-1}X^T)Y$.
\end{exmp}


The following theorem gives an upper bound on the convergence rate of the Bayesian regression Gibbs sampler.
\begin{thm}\label{thm:bayesianreggibbs}
	For two copies of the Bayesian regression Gibbs sampler, $(\beta_n,\sigma_n)$ and $(\beta'_n,\sigma^{'2}_n)$, defined in Example \ref{ex:bayesianreggibbs}, 
	\begin{align}\label{eqn:bayesianreggibbseqn}
		\norm{\L(\beta_{n},\sigma_{n})-\L(\beta'_{n},\sigma^{'2}_{n})}\leq K E[\abs{\sigma^2_0-\sigma^{'2}_0}] \left(\frac{p}{k+p-2}\right)^{n-1}
	\end{align}
	where $K= \frac{(C/2)^{\frac{k+2p}{2}}}{\Gamma(\frac{k+2p}{2})}\left(\frac{k+2p+2}{C}\right)^{\frac{k+2p}{2}+1}e^{-\frac{k+2p+2}{2}}$.
\end{thm}

In Theorem 3.1 of \cite{raj}, it was shown than for the equivalent example and some $0<M_1\leq M_2$, which are not specified, $$M_1\left(\frac{p}{k+p-2}\right)^n\leq \norm{\L(\beta_n,\sigma_n)-\pi}\leq M_2\left(\frac{p}{k+p-2}\right)^n$$
This means that the bound derived from the Corollary \ref{cor:oneshotlipschitz} is sharp up to a constant. The primary difference between Theorem 3.1 in \cite{raj} and the bound in Theorem \ref{thm:bayesianreggibbs} is that the latter provides explicit values for the constant, $M_2$ and as a result, numerical upper bounds can be calculated.

Before proving the Theorem \ref{thm:bayesianreggibbs}, we present some lemmas.

\begin{lem}\label{lem:bayesregdeinit}
	The variable $\sigma^2_{n}$ can be written as a random-coefficient autoregressive process, $\sigma^2_{n}=X_nY_n\sigma^2_{n-1}+Y_n$ where \label{eqn:xy}
	$X_n\sim \Gamma\left(\frac{p}{2}, \frac{C}{2}\right)$ and $Y_n\sim \Gamma^{-1}\left(\frac{k+p}{2}, \frac{C}{2}\right)$.
	And so, $\norm{\L(\beta_n,\sigma^2_n)-\L(\beta'_n, \sigma^{'2}_n)}\leq \norm{\L(\sigma^2_n)-\L(\sigma^{'2}_n)}$.
	
	The proof can be found in \ref{proof:bayesregdeinit}. Note that $\Gamma(\alpha,\beta)$ represents the gamma distribution and $\Gamma^{-1}(\alpha,\beta)$ represents the inverse gamma distribution.
\end{lem}

%


\begin{lem}[Contraction condition]\label{lem:bayesregcontr}
	The Bayesian regression Gibbs sampler satisfies the contraction condition with $D=\left(\frac{p}{k+p-2}\right)$. The proof can be found in \ref{proof:bayesregcontr}.
\end{lem}

\begin{lem}[Attributes of the conditional density $\theta_{2,n}\mid \theta_{1,n}$]\label{lem:bayesregconddens}
	For the Bayesian regression Gibbs sampler, $\theta_{2,n}\mid \theta_{1,n}$ has a continuous density, $M=1$ and $K = \frac{(C/2)^{\frac{k+2p}{2}}}{\Gamma(\frac{k+2p}{2})}\left(\frac{k+2p+2}{C}\right)^{\frac{k+2p}{2}+1}e^{-\frac{k+2p+2}{2}}$.
The proof can be found in \ref{proof:bayesregconddens}.
\end{lem}

Given the above lemmas, the proof of Theorem \ref{thm:bayesianreggibbs} is straightforward when the Sideways Theorem is applied.

\begin{proof}[Proof of Theorem \ref{thm:bayesianreggibbs}]
	Let $n\geq 0$.
	$$\norm{\L(\beta_{n},\sigma^2_{n})-\L(\beta'_{n}, \sigma^{'2}_{n})} \leq \norm{\L(\sigma^2_{n})-\L(\sigma^{'2}_{n})} \leq K E[\abs{\sigma^2_0-\sigma^{'2}_0}] \left(\frac{p}{k+p-2}\right)^{n-1}$$
where $K$ is defined in Lemma \ref{lem:bayesregconddens}. Lemma \ref{lem:bayesregdeinit} implies the first inequality. The second inequality is a result of Corollary \ref{cor:oneshotlipschitz}, which is satisfied because of the contraction condition (Lemma \ref{lem:bayesregcontr}) and the properties of the conditional density $\theta_{2,n}\mid \theta_{1,n}$ (Lemma \ref{lem:bayesregconddens}).

%
\end{proof}

\begin{numexmp}[Application of the Bayesian regression Gibbs sampler]
	Suppose that we are interested in evaluating the delay in getting a PhD ($Y$), based on age, age squared, sex and whether the student has a child at home ($X$). For more information on this problem see \cite{phddelay, phddelaydata}. We want to find the upper bound on the total variation distance for a Bayesian regression Gibbs sampler fitted to this model. In this case, there are 333 observed values ($k=333$) and 4 covariates ($p=4$). Using the notation from Theorem \ref{thm:bayesianreggibbs}, $K=0.0682$. Further suppose we are interested in evaluating the upper bound between two copies of the Markov chain $X_n,X'_n$ such that $\sigma_0^2=1$ and $\sigma^{'2}_0=1001$. Then,
	\begin{align}
		\norm{\L(\beta_{n},\sigma_{n})-\L(\beta '_{n},\sigma '_{n})}\leq 68.16454 \left(0.0119403\right)^{n-1}
	\end{align}
	After 3 iterations, the total variation distance between the two chains will be less than 0.01.
\end{numexmp}

\subsection{Bayesian location model Gibbs sampler}
\begin{exmp}[Bayesian location model Gibbs sampler]\label{ex:anothergibbssampler}
Suppose that we are given data points $Y_1,\ldots, Y_J \sim N(\mu, \tau^{-1})$ where $\mu, \tau^{-1}$ are unknown and $J\geq 3$. Let $\mu, \tau^{-1}$ have flat priors on $\mathbb{R}$ and $\mathbb{R}_+$. The Gibbs algorithm is based on the conditional posterior distributions of $\mu, \tin$, which are defined as follows.

\begin{itemize}
	\item $\mu_{n+1} = \bar{y} + Z_{n+1}/\sqrt{J \tau_{n}}$
	\item $\tau^{-1}_{n+1} = \frac{\frac{S}{2}+\frac{J}{2}(\bar{y}-\mu_{n+1})^2}{G_{n+1}}$
\end{itemize}
Where $Z_{n}\sim N(0,1)$ and $G_{n}\sim \Gamma(\frac{J+2}{2},1)$ are independent and $S=\sumn (y_i-\bar{y})^2$. 
\end{exmp}

The following theorem gives an upper bound on the convergence rate of the Bayesian location model Gibbs sampler. 
\begin{thm}\label{thm:anothergibbssampler}
	For two copies of the Bayesian location model Gibbs sampler Example \ref{ex:anothergibbssampler}, 
	\begin{align}\label{eqn:anothergibbssamplereqn}
		\norm{\L(\mu_{n},\tinv_{n})-\L(\mu'_{n}, \tau^{'-1}_{n})}&\leq K E[\abs{\tinv_0-\tau^{'-1}_0}] \left(\frac{1}{J}\right)^{n-1}
	\end{align}
	where $K= \frac{(S/2)^{\frac{J-1}{2}}}{\Gamma(\frac{J-1}{2})}\left(\frac{S}{J+1}\right)^{-\frac{J-3}{2}}e^{-\frac{J+1}{2}}$.
\end{thm}
This bound compares to the one derived in Section 6 of \cite{oneshot} which states that,
\begin{align*}
	\norm{\L(\mu_n,\tinv_n)-\L(\mu'_n, \tau^{'-1}_n)}
	&\leq\left(\frac{J}{2}+1\right) E[\abs{\tinv_0-\tau^{'-1}_0}] \left(\frac{1}{J}\right)^n
\end{align*}
Both bounds return the same geometric rate of convergence. However, the magnitude of constant $K$ is difficult to compare against $\left(\frac{J}{2}+1\right)$ without knowing $S$. Note that the bound derived from Corollary \ref{cor:oneshotlipschitz} is calculated in a systematic way.

Before proving Theorem \ref{thm:anothergibbssampler}, we present some lemmas.

\begin{lem}\label{lem:locmodeldeinit}
	The variable $\tinv_{n}$ can be written as a random-coefficient autoregressive process, $\tinv_{n}=X_nY_n\tinv_{n-1}+Y_n$,
	where $X_n\sim \Gamma\left(\frac{1}{2}, \frac{S}{2}\right)$ and $Y_n\sim \Gamma^{-1}\left(\frac{J+2}{2}, \frac{S}{2}\right)$. And so, $\norm{\L(\mu_n,\tau^{-1}_n)-\L(\mu'_n, \tau^{'-1}_n)}\leq \norm{\L(\tinv_{n})-\L(\tau^{'-1}_n)}$.
	The proof can be found in \ref{proof:locmodeldeinit}.
\end{lem}

\begin{lem}[Contraction condition]\label{lem:locmodelcontr}
	The Bayesian location model Gibbs sampler satisfies the contraction condition with $D=\frac{1}{J}$. The proof can be found in \ref{proof:locmodelcontr}.
\end{lem}

\begin{lem}[Attributes of the conditional density $\theta_{2,n}\mid \theta_{1,n}$]\label{lem:locmodelconddens}
	For the Bayesian location model Gibbs sampler, $\theta_{2,n}\mid \theta_{1,n}$ has a continuous density, $M=1$ and 
	\begin{align}\label{eqn:K}
		K &= \frac{(S/2)^{\frac{J-1}{2}}}{\Gamma(\frac{J-1}{2})}\left(\frac{S}{J+1}\right)^{-\frac{J-3}{2}}e^{-\frac{J+1}{2}}
	\end{align}
	The proof can be found in \ref{proof:locmodelconddens}.
\end{lem}

Given the above lemmas, the proof of Theorem \ref{thm:anothergibbssampler} is straightforward when the Sideways Theorem is applied.

\begin{proof}[Proof of Theorem \ref{thm:anothergibbssampler}]
	Note that
	$$\norm{\L(\mu_{n},\tinv_{n})-\L(\mu'_{n}, \tau^{-1'}_{n})} \leq \norm{\L(\tinv_{n})-\L(\tau^{-1'}_{n})} \leq K E[\abs{\tinv_0-\tau^{-1'}_0}] \left(\frac{1}{J}\right)^{n-1}$$
	where $K$ is defined in Lemma \ref{lem:locmodelconddens}. The first and second inequality are a result of Lemma \ref{lem:locmodeldeinit} and Corollary  \ref{cor:oneshotlipschitz}, respectively. Corollary  \ref{cor:oneshotlipschitz} is satisfied because of the contraction condition (Lemma \ref{lem:locmodelcontr}) and the properties of the conditional density $\theta_{2,n}\mid\theta_{1,n}$ (Lemma \ref{lem:locmodelconddens}).
\end{proof}

\begin{numexmp}[Application of Bayesian location model Gibbs sampler]\label{numex:locmodel}
	Suppose that we are given the girth in inches of a sample of trees (see the \texttt{trees} dataset in \texttt{R}), $Y_1,
	\ldots, Y_{31}\sim N(\mu,\tinv)$, where $\mu, \tinv$ are unknown. We want to find the upper bound on the total variation distance for the Gibbs sampler model applied to this problem. In this case the number of datapoints is 31 ($J=31$) and using the notation from Theorem \ref{thm:anothergibbssampler}, $K=13.74027$. Further, suppose that we are interested in evaluating the upper bound between a Markov chain $(\mu_{n},\tinv_{n})$ with initial value $\tinv_0=1$ and the corresponding stationary Markov chain, which is denoted as $(\mu_{\infty},\tinv_{\infty})$.

	By Lemma \ref{lem:locmodeldriftval} a drift function exists. 
	
	\begin{lem}\label{lem:locmodeldriftval} For Numerical Example \ref{numex:locmodel}, the following drift condition holds, 
		$$E[(\tinv_n+0.5248723)^2\mid \tinv_{n-1}]\leq 0.6583702 (\tinv_{n-1}+0.5248723)^2 +106.3874$$
		The proof can be found in \ref{proof:locmodeldriftval}.
	\end{lem}
	
	So by lemma \ref{lem:locmodelexpdist}, 
	\begin{equation}\label{eq:stationarynumbd}
		E[\abs{\tinv_{\infty}-\tinv_0}]\leq 18.12198
	\end{equation} Combining this with Theorem \ref{thm:anothergibbssampler},
	\begin{align*}
		\norm{\L(\mu_{n},\tinv_{n})-\L(\mu_{\infty},\tau^{-1}_{\infty})}&\leq 13.74027\times 18.12198 \left(\frac{1}{31}\right)^{n-1}= 249 \left(\frac{1}{31}\right)^{n-1}
	\end{align*}
	
	After $4$ iterations, the total variation distance between the two chains will be less than $0.01$.
	This bound compares to the bound derived in \cite{oneshot}, which, combined with equation \ref{eq:stationarynumbd}, states that $\norm{\L(\mu_n,\tinv_n)-\L(\mu'_n, \tau^{'-1}_n)}
	\leq 299 \left(\frac{1}{31}\right)^n$.
\end{numexmp}

\subsection{Autoregressive normal process}
\begin{exmp}[Autoregressive normal process in $\mathbb{R}$]\label{ex:arlipschitz}
	
	Let $\{X_n\}_{n\geq 1}\in \mathbb{R}$ be an autoregressive normal process. Then for i.i.d. $Z_n\sim N(0,1)$,
	$$X_{n}=\frac{1}{2}X_{n-1}+\sqrt{\frac{3}{4}}Z_n$$
	In this case $\theta_{1,n}=\frac{1}{2}$ and $\theta_{2,n}=\sqrt{\frac{3}{4}}Z_n$. The density of $\theta_{2,n}$ is continuous and uni-modal and $K=\sqrt{\frac{2}{3\pi}}$.
	By Corollary \ref{cor:oneshotlipschitz},
	\begin{equation}\label{eqn:arnormupperbound}
		\norm{\L(X_{n})-\L(X'_{n})}\leq \sqrt{\frac{2}{3\pi}} E[\abs{X_0-X'_0}]\left(\frac{1}{2}\right)^{n-1}
	\end{equation}
	It is known that the geometric rate of convergence for the autoregressive normal process is $1/2$ \cite{wasmeth}, so once again the Sideways Theorem \ref{thm:oneshotlipschitz} generates tight geometric convergence rates up to a constant.
	
	When comparing the upper bound with the actual total variation distance, note that if $X_0=x_0$ is known, $X_n\sim N(\frac{x_0}{2^n}, 1-\frac{1}{4^n})$. Thus, the total variation distance between two copies of an autoregressive normal process $X_n, X'_n$ where the initial values are known, $X_0=x_0$ and $X'_0=x'_0$, is as follows (see Section 2 of \cite{oneshot}),
	\begin{equation*}
		\norm{\L(X_n)-\L(X'_n)}=1-2\Phi \left(-\frac{\abs{x_0-x'_0}}{2^{n+1}\sqrt{1-\frac{1}{4^n}}}\right) 
	\end{equation*}

Figure \ref{fig:arnormprocessuppervsactual} shows how the upper bound for the autoregressive normal process using equation \ref{eqn:arnormupperbound} compares to the actual total variation distance when $x_0=0$ and $x'_0=1$. The total variation is less than 0.01 after 6 iteration and the upper bound on the total variation is less than 0.01 after 7 iterations.
 
 \begin{figure}
 	\centering
 	\includegraphics[width=0.7\linewidth]{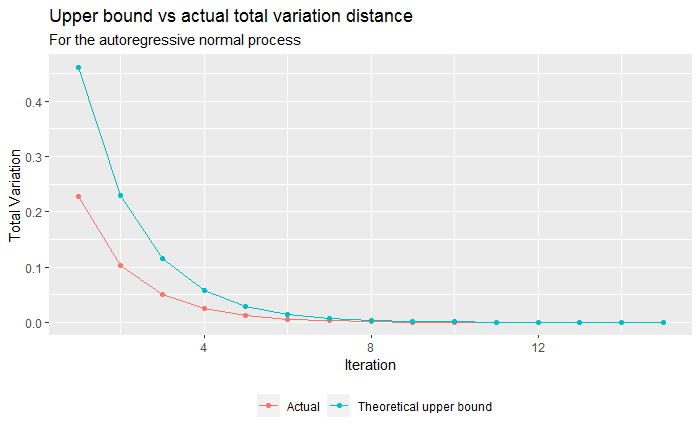}
 	\caption{This figure compares the actual value of $\norm{\L(X_n)-\L(X'_n)}$ against the upper bound derived from the Sideways Theorem \ref{thm:oneshotlipschitz}, (Equation \ref{eqn:arnormupperbound}) when $X_n, X'_n$ are two copies of the autoregressive normal process (i.e., $X_n=\frac{1}{2}X_{n-1}+\sqrt{\frac{3}{4}}Z_n, Z_n\sim N(0,1)$) and $x_0=0$, $x'_0=1$.}
 	\label{fig:arnormprocessuppervsactual}
 \end{figure}
 
\end{exmp}

In the following section we extend the above example to higher dimensions.

\subsection{Processes in $\mathbb{R}^d$}\label{subsec:rdex}

Next we extend the autoregressive normal process as defined above to $\mathbb{R}^d$. 
To do so, we apply Proposition \ref{prop:indcoord} to an autoregressive normal process in $\mathbb{R}^d$ with independent coordinates, Example \ref{ex:arnormind}, and non-independent coordinates, Example \ref{ex:arnormdep}. 

\begin{exmp}[Autoregressive normal process in $\mathbb{R}^d$ with independent coordinates]\label{ex:arnormind}
	
	Let $\{\vecx_n\}_{n\geq 1}\in \mathbb{R}^d$ be an autoregressive normal process with independent coordinates. Then for i.i.d. $\vec{Z}_n\sim N(\vec{0},I_d)$,
	$$\vecx_{n}=\frac{1}{2}\vecx_{n-1}+\sqrt{\frac{3}{4}}\vec{Z}_n$$
	And if $i\neq j$, then $Z_{i,n}\indep Z_{j,n}$. Further, $X_{i,n}=\frac{1}{2}X_{i,n-1}+\sqrt{\frac{3}{4}}Z_{i,n}$ for $i\in \{1,\ldots,d\}$ and so by Example \ref{ex:arlipschitz}, 
	$$\norm{\L(X_{i,n})-\L(X'_{i,n})}\leq \sqrt{\frac{2}{3\pi}} E[\abs{X_{i,0}-X'_{i,0}}]\left(\frac{1}{2}\right)^{n-1}$$
	Since each coordinate is independent and bounded above by the same value, Proposition \ref{prop:indcoord} implies that
	$$\norm{\L(\vecx_{n})-\L(\vecx'_{n})}\leq d \sqrt{\frac{2}{3\pi}} \sup_{0\leq i\leq d}E[\abs{X_{i,0}-X'_{i,0}}]\left(\frac{1}{2}\right)^{n-1}$$
	
	Again, it is known that the geometric rate of convergence for the autoregressive normal process in $\mathbb{R}^d$ is $1/2$ \cite{wasmeth}. 
	
	Finally, to apply numbers to this example, suppose that $\vecx_n, \vecx '_n \in \mathbb{R}^{100}$ and the initial values of this process are $\vecx_0 = (1,\ldots, 1)$ and $\vecx '_0 = (0,\ldots, 0)$. The total variation distance would be bounded above with $\norm{\L(\vecx_{n+1})-\L(\vecx'_{n+1})}\leq 100 \sqrt{\frac{2}{3\pi}}\left(\frac{1}{2}\right)^n$. This means that at $14$ iterations the total variation distance would be less than $0.01$.
\end{exmp}

The following example is a more general version of the above, where $X_n$ is a general auto regressive normal process in $\mathbb{R}^d$.
\begin{exmp}[Autoregressive normal process in $\mathbb{R}^d$]\label{ex:arnormdep}
	The random vector $\{\vecx_n\}_{n\geq 1}\in \mathbb{R}^d$ is an autoregressive normal process if for matrix $A$ and random vector $\vec{W}_n\sim N(\vec{0},\Sigma_d^2)$ ($\Sigma_d^2$ is a positive semi-definite matrix)
	$$\vecx_{n}=A\vecx_{n-1}+\vec{W}_n$$
\end{exmp}

\begin{thm}\label{thm:arnormd}
	Suppose that $A$ is a diagonalizable matrix. Then for two copies, $\vecx_n, \vecx '_n\in \mathbb{R}^d$, of the autoregressive normal process defined in Example \ref{ex:arnormdep},
	\begin{equation}\label{eqn:arnormd}
		\norm{\L(\vecx_n)-\L(\vecx '_n)}\leq \sqrt{\frac{d}{2\pi}} \norm{\Sigma^{-1}_d}_2 \cdot \norm{ P}_2 \norm{P^{-1}}_2 E[\norm{\vecx_{0}-\vecx '_{0}}_2] \max_{1\leq i\leq d}\abs{\lambda_i}^n
	\end{equation}
where $A=P D P^{-1}$ with $D$ as the corresponding diagonal matrix, $\lambda_i$ is the $i$th eigenvalue of $A$ and $\norm{\cdot}_2$ denotes the Frobenius norm.
The proof can be found in \ref{proof:arnormd}, which uses a modified version of the Sideways Theorem.
\end{thm}

\begin{numexmp}[Application of the autoregressive normal process in $\mathbb{R}^d$]
	To apply numbers to this example, suppose that $\vecx_n, \vecx '_n \in \mathbb{R}^{100}$ are two copies of the following process $\vecx_n = A \vecx_n + \vec{Z}_n, \vec{Z}_n\sim N(0, A)$ where 
	\begin{equation*}
		A = 
	\begin{pmatrix}
	\frac{1}{2} & \frac{1}{8} & 0 & \cdots & 0 & 0 \\
	\frac{1}{8} & \frac{1}{2} & \frac{1}{8} & \cdots & 0 & 0 \\
	\vdots  & \vdots  & \vdots & \ddots  & \vdots & \vdots  \\
	0 & 0 & 0 & \cdots & \frac{1}{8} & \frac{1}{2}\\
\end{pmatrix}
	\end{equation*}
	 and the initial values of this process are $\vecx_0 = (1,\ldots, 1)$ and $\vecx '_0 = (0,\ldots, 0)$. The total variation distance would be bounded above with $\norm{\L(\vecx_{n})-\L(\vecx'_{n})}\leq 98782.31\left(0.7498791 \right)^n$. This means that after $56$ iterations the total variation distance would be less than $0.01$.
\end{numexmp}

\section{Autoregressive conditional heteroscedastic processes}\label{sec:linproc}

In this section we look at bounding the total variation distance between two copies of an ARCH process.

\begin{defn}[Autoregressive conditional heteroscedastic (ARCH) process]\label{def:ARCH}
	The sequence $\{X_n\}_{n\geq 1}$ is an ARCH process if for $g:\mathbb{R}^2\to \mathbb{R}$
	\begin{equation}
		X_n = g(\theta_{1,n}, X_{n-1})\theta_{2,n}
	\end{equation}
	where $(\theta_{1,n},\theta_{2,n})\in \mathbb{R}^2$ are random variables and $(\theta_{1,n},\theta_{2,n})\indep (\theta_{1,m},\theta_{2,m})$ when $n\neq m$.
\end{defn}

\subsection{Application to the LARCH model}

\begin{exmp}[Linear ARCH process]\label{ex:lineararch}
	Let $\{X_n\}_{n\geq 1}\in \mathbb{R}$ be a linear ARCH process. Then for i.i.d. $Z_n$ and $\beta_0,\beta_1\in \mathbb{R}$
	$$X_n = (\beta_0 +\beta_1 X_{n-1})Z_n$$
	See Section 7.3.3 of \cite{timeseries} for more details on this model.
\end{exmp}

The following theorem provides an upper bound on the convergence rate of two copies of a LARCH process.

\begin{thm}\label{thm:lineararch}
	Let $\{X_n\}_{n\geq 1}\in \mathbb{R}$ and $\{X'_n\}_{n\geq 1}\in \mathbb{R}$ be two copies of the linear ARCH process. Suppose that, 
	\begin{itemize}
		\item $\beta_0,\beta_1>0$ and $Z_n>0$ a.s.
		\item the density of $\log(Z_0)$ is bounded above, has at most $M$ local maxima and minima, and is continuous. 
	\end{itemize}
	Then, the process is geometrically ergodic if $\beta_1 E[\abs{Z_0}]<1$ and an upper bound on the total variation distance between the two processes is,
	\begin{equation}\label{eqn:larch}
		\norm{\L(X_{n})-\L(X'_{n})} \leq \frac{\beta_1(M+1)}{2\beta_0}\sup_x e^x f_{Z_n}(e^x)D^{n-1}E[\abs{X_0-X'_0}]
	\end{equation}
	Where $D=\beta_1 E[Z_0]$
\end{thm}
Lemma 7.3.2 of \cite{timeseries} says that if $\beta_1E[\abs{Z_0}]<1$, then a stationary distribution exists. This theorem makes an even stronger assertion that under some additional assumptions, the process will also be geometrically ergodic with geometric convergence rate $D=\beta_1E[\abs{Z_0}]<1$.

Before proving Theorem \ref{thm:lineararch}, we present some lemmas.

\begin{lem}[Contraction condition]\label{lem:larchcontr}
	The LARCH process satisfies the contraction condition if $D=\beta_1 E[Z_0]<1$.
	See Section \ref{proof:larchcontr} for a proof.
\end{lem}

\begin{lem}[Coalescing condition]\label{lem:larchcoalesc}
	Suppose that the density of $\log(Z_0)$ is bounded above, has at most $M$ local maxima and minima and is continuous. Then the LARCH process satisfies the coalescing condition
	$$\norm{\L(X_{n})-\L(X'_{n})}\leq C E[\abs{X_{n-1}- X'_{n-1}}]$$
	
	Where $n\geq 1$ and $C= \frac{\beta_1(M+1)}{2\beta_0}\sup_x e^x f_{Z_n}(e^x)$,
	See Section \ref{proof:larchcoalesc} for a proof.
\end{lem}
Note that the density of $\log(Z_0)$ is $f_{\log(Z_0)}(x) =  e^x f_{Z_0}(e^x)$.

\begin{proof}[Proof of Theorem \ref{thm:lineararch}]\label{proof:larch}
	Suppose that the assumptions in Theorem \ref{thm:lineararch} are satisfied. Then the LARCH model satisfies the contraction condition (Lemma \ref{lem:larchcontr}) and the coalescing condition (Lemma \ref{lem:larchcoalesc}). By the One-Shot Coupling Theorem \ref{thm:oneshot}, equation \ref{eqn:larch} holds.
\end{proof}

\begin{numexmp}
	We find the convergence rate of Example 10.3.1 of \cite{intrototimeseries}, which is of the form,
	$$X_n^2 = (1 + 0.5X_{n-1}^2)Z_n^2$$
	Where $Z^2_n\sim \chi^2(1)$. Further let $X_0=0.1$ and $X'_0=1.1$. The density of $\log(Z_n^2)$ is $f_{\log(Z_n^2)}(x)= (2\pi)^{-1/2}e^{(x-e^x)/2}$ and so, $\sup f_{\log(Z_n^2)}(x) = (2\pi)^{-1/2}e^{(0-e^0)/2}=\frac{1}{\sqrt{2\pi e}}$. The density of $\log(Z_n^2)$ is also unimodal, so $M=1$. By Theorem \ref{thm:lineararch} an upper bound on the total variation distance is
	\begin{equation}
		\norm{\L(X_{n})-\L(X'_{n})}\leq \frac{1}{\sqrt{8\pi e}} 0.5^{n-1}
	\end{equation}
	After 3 iterations the total variation distance is less than 0.01. In comparison, Figure \ref{fig:larch} shows how the bound compares to a simulated estimate of the total variation distance for this process.
	\begin{figure}
		\centering
		\includegraphics[width=0.7\linewidth]{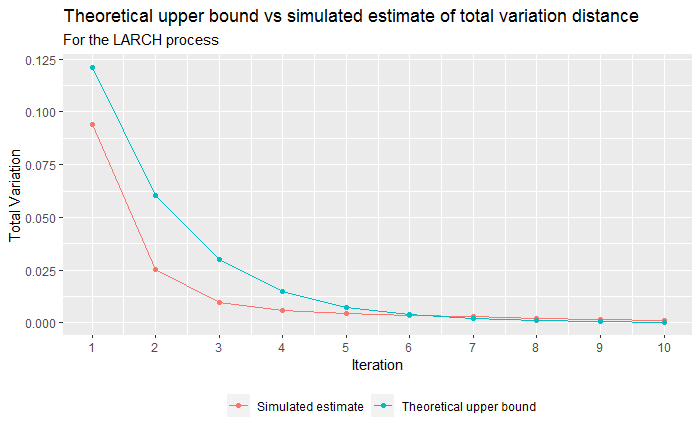}
		\caption{This figure compares a simulated approximation of $\norm{\L(X^2_n)-\L(X^{'2}_n)}$ against the upper bound (Equation \ref{eqn:larch}). $X^2_n, X^{'2}_n$ are two copies of the LARCH process (i.e., $X_n^2 = (1 + 0.5X_{n-1}^2)Z_n^2$ and $Z_n^2\sim \chi^2(1)$) and $X^2_0=0.1, X^{'2}_0=1.1$. To simulate total variation, 10 million simulations were run with bin length=0.01 for the estimated density function.}
		\label{fig:larch}
	\end{figure}
\end{numexmp}

\subsection{Application to the asymmetric ARCH model}

\begin{exmp}[Asymmetric ARCH process]\label{ex:asymetricarch}
	Let $\{X_n\}_{n}{n\geq 1}\in \mathbb{R}$ be an asymmetric ARCH process. Then for i.i.d. $Z_n$
	$$X_n = \sqrt{(aX_{n-1}+b)^2+c^2}Z_n$$
	Where $a, b,c \in \mathbb{R}$. See Exercise 4.1 of \cite{timeseries} for more details on this process.
\end{exmp}

The following theorem provides an upper bound on the convergence rate of two copies of an asymmetric ARCH process.

\begin{thm}\label{thm:asymmetricarch}
	Let $\{X_n\}_{n\geq 1}\in \mathbb{R}$ and $\{X'_n\}_{n\geq 1}\in \mathbb{R}$ be two copies of the asymmetric ARCH process defined in Example \ref{ex:asymetricarch}. Suppose further that the density of $Z_n$ is centred at 0 and is monotonically decreasing around zero (i.e., $\pi(x)\geq \pi(y)$ if $\abs{x}<\abs{y}$.). Then, the process is geometrically ergodic if $\abs{a}E[\abs{Z_0}]<1$ and an upper bound on the total variation distance between the two processes is
	\begin{equation}\label{eqn:asymmetricarch}
		\norm{\L(X_{n})-\L(X'_{n})}\leq \frac{\abs{a}}{c}D^{n-1}E[\abs{X_0-X'_0}]
	\end{equation}
	Where $D=\abs{a}E[\abs{Z_0}]$
\end{thm}
Exercise 4.1 part 1 of \cite{timeseries} states that the process has a stationary solution if $D=\abs{a}E[\abs{Z_0}]<1$. Theorem \ref{thm:asymmetricarch} shows that under certain additional assumptions on $Z_n$ the process will also be geometrically ergodic with a specified quantitative bound.

Before proving Theorem \ref{thm:asymmetricarch}, we present some lemmas.

\begin{lem}[Contraction condition]\label{lem:asymarchcontraction}
	The asymmetric ARCH process satisfies the contraction condition if $D=\abs{a}E[\abs{Z_0}]<1$.
	See Section \ref{proof:asymarchcontraction} for a proof.
\end{lem}

\begin{lem}[Coalescing condition]\label{lem:asymarchcoalesc}
	Suppose that the density of $Z_n$ is centred at 0 and is monotonically decreasing around zero. Then, the asymmetric ARCH process satisfies the coalescing condition
	$$\norm{\L(X_{n})-\L(X'_{n})}\leq C E[\abs{X_{n-1}-X'_{n-1}}]$$
	where $n\geq 1$ and $C=\frac{\abs{a}}{c}$.
	See Section \ref{proof:asymarchcoalesc} for a proof.
\end{lem}

\begin{proof}[Proof of Theorem \ref{thm:asymmetricarch}]\label{proof:asymmetricarch}
	Suppose that the assumptions in Theorem \ref{thm:asymmetricarch} are satisfied. Then the asymmetric ARCH model satisfies the contraction condition (Lemma \ref{lem:asymarchcontraction}) and the coalescing condition (Lemma \ref{lem:asymarchcoalesc}). By the One-Shot Coupling Theorem \ref{thm:oneshot}, equation \ref{eqn:asymmetricarch} holds.
\end{proof}

\begin{numexmp}
	Suppose $a=0.5, b=3, c=5$, $Z_n\sim N(0,1)$ and $X_0=0, X'_0=5$. Then by Jensen's inequality, $D=0.5E[\abs{Z_0}] \leq 0.5E[Z_0^2]^{1/2}=0.5$ and so by Theorem \ref{thm:asymmetricarch}
	\begin{equation}\label{eqn:numasymarch}
		\norm{\L(X_n)-\L(X'_n)}\leq \frac{0.5}{5}\times 5\times 0.5^{n-1} = 0.5^n
	\end{equation}
	So, by iteration $n=7$, the total variation will be less than $0.01$.
	
	In comparison, Figure \ref{fig:asymarch} shows how the bound compares to a simulated estimate of the total variation distance for this process.

	\begin{figure}
		\centering
		\includegraphics[width=0.7\linewidth]{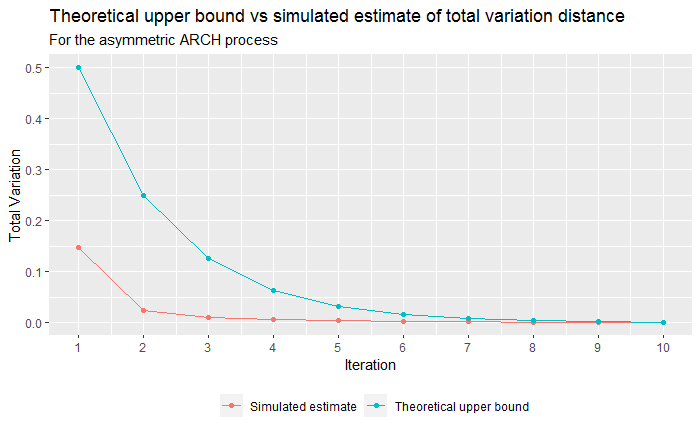}
		\caption{This figure compares a simulated approximation of $\norm{\L(X_n)-\L(X'_n)}$ against the upper bound (Equation \ref{eqn:numasymarch}). $X_n, X'_n$ are two copies of the asymmetric process (i.e., $X_n = \sqrt{(0.5X_{n-1}+3)^2+5^2}Z_n, Z_n\sim N(0,1)$) and $x_0=0$, $x'_0=5$.
		To simulate total variation, 10 million simulations were run with bin length=0.01 for the estimated density function.}
		\label{fig:asymarch}
	\end{figure}
	
\end{numexmp}

\subsection{Application to the GARCH(1,1) model}
\begin{exmp}[GARCH(1,1) process]\label{ex:garch}
	Let $\{X_n\}_{n\geq 1}\in \mathbb{R}$ be a GARCH(1,1) process. Then for i.i.d. $Z_n$
	$$X_n = \sigma_nZ_n$$
	where for $\alpha, \beta, \gamma \in \mathbb{R}$,
	$$\sigma_n^2 = \alpha^2 +\beta^2X^2_{n-1}+\gamma^2 \sigma^2_{n-1}$$
	See Section 7.3.6 of \cite{timeseries} for more details on this model.
\end{exmp}

The following theorem provides an upper bound in total variation distance between two copies of the GARCH(1,1) process.

\begin{thm}\label{thm:garch}
	Let $\{X_n\}_{n\geq 1}\in \mathbb{R}$ and $\{X'_n\}_{n\geq 1}\in \mathbb{R}$ be two copies of the GARCH process defined in Example \ref{ex:garch}. Suppose that the density of $Z_n$ is centered at 0 and is monotonically decreasing around zero. Then, the process is geometrically ergodic if $\beta^2E[\abs{Z^2_0}]+\gamma^2<1$. Further suppose that $x_0, x'_0, \sigma^2_0$, and $\sigma^{'2}_0$ are known. Then an upper bound on the total variation distance between the two processes is
	\begin{equation}
		\norm{\L(X_{n})-\L(X'_{n})} \leq \frac{D^{n-1}}{\alpha} \sqrt{\beta^2 \abs{x_0^2 - x_0^{'2}} + \gamma^2 \abs{\sigma^2_0 - \sigma^{'2}_0}}
	\end{equation}
	Where $D=\sqrt{\beta^2E[Z^2_0]+\gamma^2}$
\end{thm}
Before proving Theorem \ref{thm:garch}, we present some lemmas.

\begin{lem}[Contraction condition]\label{lem:garchcontraction}
	The GARCH(1,1) process satisfies the contraction condition if $D=\sqrt{\beta^2E[Z_0^2]+ \gamma^2}<1$
	See Section \ref{proof:garchcontraction} for a proof.
\end{lem}

\begin{lem}[Coalescing condition]\label{lem:garchcoalescing}
	Suppose that the density of $Z_n$ is centred at 0 and is monotonically decreasing around zero. Then the GARCH(1,1) process satisfies the coalescing condition, $$\norm{\L(X_{n})-\L(X'_{n})} \leq \frac{D}{\alpha E[\abs{Z_{0}}]} E[\abs{X_{n-1}-X'_{n-1}}]$$
	For $n\geq 2$,  $D=\sqrt{\beta^2 E[Z^2_0]+\gamma^2}$.
	See Section \ref{proof:garchcoalescing} for a proof.
\end{lem}

\begin{lem}[Initial condition]\label{lem:garchinitial}
	Suppose that we know $\sigma_0^2, \sigma_0^{'2}$ and $X_0, X'_0$, then
	$$E[\abs{X_1 - X'_1}]\leq \sqrt{\beta^2 \abs{X_0^2 - X_0^{'2}} + \gamma^2 \abs{\sigma^2_0 - \sigma^{'2}_0}}E[\abs{Z_0}]$$
	See Section \ref{proof:garchinitial} for a proof.
\end{lem}

\begin{proof}[Proof of Theorem \ref{thm:garch}]\label{proof:corgarch}
	Suppose that the assumptions in Theorem \ref{thm:garch} are satisfied and let $n\geq 2$. Then the GARCH(1,1) model satisfies the contraction condition (Lemma \ref{lem:garchcontraction}) and the coalescing condition (Lemma \ref{lem:garchcoalescing}). Thus by the One-Shot Coupling Theorem \ref{thm:oneshot},
	$$\norm{\L(X_{n})-\L(X'_{n})} \leq \frac{D}{\alpha E[\abs{Z_{0}}]} D^{n-2}E[\abs{X_{1}-X'_{1}}]$$
	Further, by Lemma \ref{lem:garchinitial} when the initial values $\sigma_0^2, \sigma_0^{'2},x_0, x'_0$ are known,
	 $$\norm{\L(X_{n})-\L(X'_{n})} \leq \frac{D^{n-1}}{\alpha} \sqrt{\beta^2 \abs{X_0^2 - X_0^{'2}} + \gamma^2 \abs{\sigma^2_0 - \sigma^{'2}_0}}$$
	 where $D=\sqrt{\beta^2 E[Z^2_0]+\gamma^2}$
\end{proof}

\begin{numexmp}
	In Example 10.3.2 of \cite{intrototimeseries} a GARCH(1,1) model is applied for the daily returns of the Dow Jones Industrial Index between between July 1997 and April 1999. Let 
	$$X_n=\sigma_n Z_n = \text{excess daily return of the Dow Jones Industrial Index at time } n$$ The following is the fitted GARCH volatility estimates when $Z_n\sim N(0,1)$,
	$$\sigma_n^2 = 0.13000 + 0.1266 X^2_{n-1} +0.7922 \sigma^2_{n-1}$$
	Suppose that we want to find the total variation of the fitted process with varying initial values representing two market states, $X_0=0.1, \sigma_0=0.01$ and $X'_0=-0.1, \sigma'_0=0.1$ Then by Theorem \ref{thm:garch}, 
	\begin{equation}\label{eqn:numgarch}
		\norm{\L(X_{n})-\L(X'_{n})}\leq \sqrt{\frac{0.7922\abs{0.01^2-0.1^2}}{0.13}}D^{n-1} \approx 0.2456 D^{n-1} 
	\end{equation}
	Where $D=\sqrt{0.1266 + 0.7922}= \sqrt{0.9188}$
	
	By iteration 77 the total variation distance between the two processes is less than 0.01. In comparison, Figure \ref{fig:asymarch} shows how the bound compares to a simulated estimate of the total variation distance for this process. The actual total variation distance appears to be much smaller than the upper bound. 
	
	\begin{figure}
		\centering
		\includegraphics[width=0.7\linewidth]{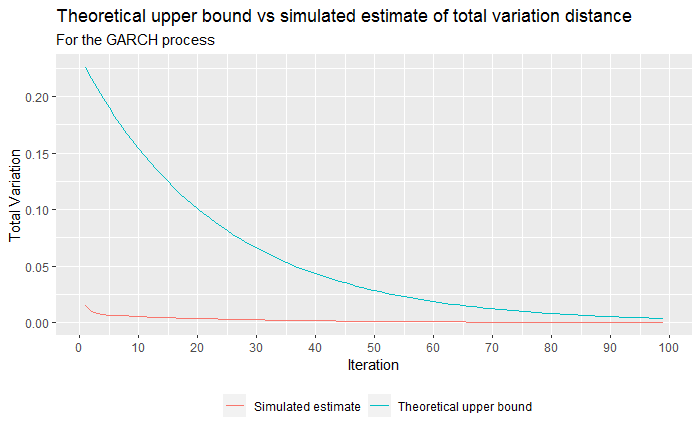}
		\caption{This figure compares a simulated approximation of $\norm{\L(X_n)-\L(X'_n)}$ against the upper bound (Equation \ref{eqn:numgarch}). $X_n, X'_n$ are two copies of the asymmetric process (i.e., $X_n = \sigma_n Z_n$ and $ \sigma_n^2 = 0.13000 + 0.1266 X^2_{n-1} +0.7922 \sigma^2_{n-1}$ and $Z_n\sim N(0,1)$) and $X_0=0.1, \sigma_0=0.01$ and $X'_0=-0.1, \sigma'_0=0.1$.
			To simulate total variation, 1 million simulations were run with bin length=0.01 for the estimated density function.}
		\label{fig:asymarch}
	\end{figure}
\end{numexmp}

\section{Acknowledgement}
We thank the referees for their many excellent comments and suggestions.

\printbibliography

@book{firstlook,
   author = {Jeffrey S. Rosenthal},
   year = {2016},
   title = {A First Look at Rigorous Probability Theory},
   publisher = {World Scientific},
   edition = {2},
   doi = {10.1142/6300}
}

@article{durmuswasconv,
   author={Alain Durmus and Éric Moulines},
   title={Quantitative bounds of convergence for geometrically ergodic Markov chain in the Wasserstein distance with application to the Metropolis Adjusted Langevin Algorithm},
   journal={Statistics and Computing},
   volume={25},
   year=2015,
   pages={5-19},
   doi = {10.1007/s11222-014-9511-z}
}

@misc{jacob,
author = {Pierre E. Jacob},
title = {Lecture notes for Couplings and Monte Carlo},
howpublished = {Available at \url{https://sites.google.com/site/pierrejacob/cmclectures?authuser=0} (2021/09/17)}
}

@misc{wasmeth,
      title={Wasserstein-based methods for convergence complexity analysis of MCMC with applications}, 
      author={Qian Qin and James P. Hobert},
      year={2020},
      eprint={1810.08826},
      archivePrefix={arXiv},
      primaryClass={math.ST}
}

@misc{geomconvrates,
      title={Geometric convergence bounds for Markov chains in Wasserstein distance based on generalized drift and contraction conditions}, 
      author={Qian Qin and James P. Hobert},
      year={2021},
      eprint={1902.02964},
      archivePrefix={arXiv},
      primaryClass={math.PR}
}

@misc{waslinearmixedmodels,
      title={Dimension free convergence rates for Gibbs samplers for Bayesian linear mixed models}, 
      author={Zhumengmeng Jin and James P. Hobert},
      year={2021},
      eprint={2103.06324},
      archivePrefix={arXiv},
      primaryClass={math.ST}
}

@article{wasandoneshot,
	 author = {Neal Madras and Denis Sezer},
	 journal = {Bernoulli},
	 number = {3},
	 pages = {882--908},
	 title = {Quantitative bounds for Markov chain convergence: Wasserstein and total variation distances},
	 volume = {16},
	 year = {2010},
	 doi = {10.2307/25735016}
}

@article{baxendale,
	author = {Peter H. Baxendale},
	title = {{Renewal theory and computable convergence rates for geometrically ergodic Markov chains}},
	volume = {15},
	journal = {The Annals of Applied Probability},
	number = {1B},
	publisher = {Institute of Mathematical Statistics},
	pages = {700 -- 738},
	year = {2005},
	doi = {10.1214/105051604000000710}
}

@phdthesis{jerison,
    title    = {The Drift and Minorization Method for Reversible Markov Chains},
    school   = {Stanford University},
    author   = {Daniel Jerison},
    year     = {2016}
}

@article{eigen,
	author = {Jeffrey S. Rosenthal},
	title = {Convergence Rates for Markov Chains},
	volume = {37},
	journal = {SIAM Review},
	number = {3},
	publisher = {Society for Industrial and Applied Mathematics},
	pages = {387-405},
	year = {1995},
	doi = {10.1137/1037083}
}

@misc{raj,
      title={MCMC-Based inference in the era of big data: a fundamental analysis of the convergence complexity of high-dimensional chains}, 
      author={Bala Rajaratnam and Doug Sparks},
      year={2015},
      eprint={1508.00947},
      archivePrefix={arXiv},
      primaryClass={math.ST}
}

@article{robrose,
  title={General state space Markov chains and MCMC algorithms},
  author={Gareth O. Roberts and Jeffrey S. Rosenthal},
  journal={Probability Surveys},
  year={2004},
  volume={1},
  pages={20-71},
  doi = {10.1214/154957804100000024}
}

@article{numsplitting,
	author = {E. Nummelin.},
	title = {A splitting technique for Harris recurrent chains},
	journal = {Z. Wahrscheinlichkeitstheorie und Verw. Geb.},
	volume = {43},
	pages = {309-318},
	year = {1978},
	doi = {10.1007/BF00534764}
}

@article{iteratedranfunc,
	author = {Persi Diaconis and David Freedman},
	title = {Iterated random functions},
	journal = {SIAM Review},
	volume = {41},
	number = {1},
	pages = {45-76},
	year = {1999},
	doi = {10.1137/S0036144598338446}
}

@misc{arconvrate,
      title={Quasi-compactness of Markov kernels on weighted-supremum spaces and geometrical ergodicity}, 
      author={Denis Guibourg and Loïc Hervé and James Ledoux},
      year={2012},
      eprint={1110.3240},
      archivePrefix={arXiv},
      primaryClass={math.PR}
}

@article{stein,
  title={Analysis of the Gibbs sampler for a model related to James-Stein estimators},
  author={Jeffrey S. Rosenthal},
  journal={Statistics and Computing},
  year={1996},
  volume={6},
  pages={269-275},
  doi = {10.1007/BF00140871}
}

@article{gibbs,
	 author = {Alison L. Gibbs and Francis Edward Su},
	 journal = {International Statistical Review / Revue Internationale de Statistique},
	 number = {3},
	 pages = {419--435},
	 publisher = {[Wiley, International Statistical Institute (ISI)]},
	 title = {On Choosing and Bounding Probability Metrics},
	 volume = {70},
	 year = {2002},
	 doi = {10.2307/1403865}
}

@article{gibbswass,
	author = {Alison L. Gibbs},
	title = {Convergence in the Wasserstein Metric for Markov Chain Monte Carlo Algorithms with Applications to Image Restoration},
	journal = {Stochastic Models},
	volume = {20},
	number = {4},
	pages = {473-492},
	year  = {2004},
	publisher = {Taylor & Francis},
	doi = {10.1081/STM-200033117}
}

@misc{jun,
      title={Complexity results for MCMC derived from quantitative bounds}, 
      author={Jun Yang and Jeffrey S. Rosenthal},
      year={2019},
      eprint={1708.00829},
      archivePrefix={arXiv},
      primaryClass={stat.CO}
}

@book{meyneandtweedie,
   author = {Sean P. Meyn and Richard L. Tweedie},
   year = {1993},
   title = {Markov Chains and Stochastic Stability},
   publisher = {Springer-Verlag},
   address = {London},
   doi = {10.1007/978-1-4471-3267-7}
}

@article{deinit,
	 author = {Gareth O. Roberts and Jeffrey S. Rosenthal},
	 journal = {Scandinavian Journal of Statistics},
	 number = {3},
	 pages = {489--504},
	 title = {Markov Chains and de-initializing processes},
	 volume = {28},
	 year = {2001},
	 doi = {10.1111/1467-9469.00250}
}

@article{oneshot,
	  title={One-shot coupling for certain stochastic recursive sequences},
	  author={Gareth O. Roberts and Jeffrey S. Rosenthal},
	  journal={Stochastic Processes and their Applications},
	  year={2002},
	  volume={99},
	  pages={195-208},
	  doi = {10.1016/S0304-4149(02)00096-0}
}

@article{liu,
	 author = {Jun S. Liu and Wing Hung Wong and Augustine Kong},
	 journal = {Biometrika},
	 number = {1},
	 pages = {27--40},
	 publisher = {[Oxford University Press, Biometrika Trust]},
	 title = {Covariance Structure of the Gibbs Sampler with Applications to the Comparisons of Estimators and Augmentation Schemes},
	 volume = {81},
	 year = {1994},
	 doi = {10.1093/biomet/81.1.27}
}

@book{analbyhist,
   author = {Ernst Hairer and Gerhard Wanner},
   year = {2008},
   title = {Analysis by Its History},
   publisher = {Springer-Verlag},
   address = {New York},
   doi = {10.1007/978-0-387-77036-9}
}

@book{billingsley,
   author = {Patrick Billingsley},
   year = {2012},
   title = {Probability and Measure: Anniversary Edition},
   publisher = {Wiley Series in Probability and Statistics},
   address = {New York}
}

@book{mode,
   author = {Peter D. Hoff},
   year = {2009},
   title = {A First Course in Bayesian Statistical Methods},
   publisher = {Springer},
   address = {New York},
   doi = {10.1007/978-0-387-92407-6}
}

@misc{maximalcoupling,
      title={Markovian Maximal Coupling of Markov Processes}, 
      author={Björn Böttcher},
      year={2017},
      eprint={1710.09654},
      archivePrefix={arXiv},
      primaryClass={math.PR}
}

@article{phddelay,
	 author = {Rens van de Schoot and Mara A. Yerkes and Jolien M. Mouw and Hans Sonneveld},
	 journal = {PLoS ONE},
	 number = {7},
	 pages = {e68839},
	 publisher = {Public Library of Science},
	 title = {What Took Them So Long? Explaining PhD Delays among Doctoral Candidates},
	 volume = {8},
	 year = {2013},
	 doi = {0.1371/journal.pone.0068839}
}

@online{phddelaydata,
  author = {Laurent Smeets and Rens van de Schoot},
  title = {R regression Bayesian (using brms)},
  year = {2019},
  url = {www.rensvandeschoot.com/tutorials/r-linear-regression-bayesian-using-brms/},
  urldate = {2021-06-03}
}

@article{oneshotnestingexmple,
	 author = {Oliver Jovanovski and Neal Madras},
	 journal = {Bernoulli},
	 number = {23},
	 pages = {603 - 625},
	 title = {Convergence rates for a hierarchical Gibbs sampler},
	 volume = {1},
	 year = {2013},
	 doi = {10.3150/15-BEJ758}
}

@article{oneshotimageexample,
	 author = {Oliver Jovanovski},
	 journal = {Statistics \& Probability Letters},
	 pages = {11-16},
	 title = {Convergence bound in total variation for an image restoration model},
	 volume = {90},
	 year = {2014},
	 doi = {10.1016/j.spl.2014.03.007}
}

@book{frobnorm,
   author = {Charu Aggarwal},
   year = {2020},
   title = {Linear Algebra and Optimization for Machine Learning: A Textbook},
   publisher = {Springer International Publishing},
   doi = {10.1007/978-3-030-40344-7}
}

@book{timeseries,
   author = {Paul Doukhan},
   year = {2018},
   title = {Stochastic Models for Time Series},
   publisher = {Springer International Publishing},
   edition = {1},
   doi = {10.1007/978-3-319-76938-7}
}

@book{markovmixing,
   author = {David A. Levin and Yuval Peres and Elizabeth L. Wilmer},
   year = {2017},
   title = {Markov Chains and Mixing Times},
   publisher = {American Mathematical Society},
   edition = {2},
   doi = {10.1090/mbk/107}
}

@book{intrototimeseries,
   author = {Peter J. Brockwell and Richard A. Davis},
   year = {2002},
   title = {Introduction to Time Series and Forecasting},
   publisher = {Springer Texts in Statistics},
   edition = {2},
   doi = {10.1007/978-3-319-29854-2}
}

@article{oneshotnsphereexample,
	author = {Natesh S. Pillai and Aaron Smith},
	title = {{Kac’s walk on $n$-sphere mixes in $n\log n$ steps}},
	volume = {27},
	journal = {The Annals of Applied Probability},
	number = {1},
	publisher = {Institute of Mathematical Statistics},
	pages = {631 -- 650},
	year = {2017},
	doi = {10.1214/16-AAP1214}
}

@article{rose,
	author = {Jeffrey S. Rosenthal},
	title = {Minorization Conditions and Convergence Rates for Markov Chain Monte Carlo},
	volume = {90},
	journal = {Journal of the American Statistical Association},
	number = {430},
	pages = {558-566},
	year = {1995},
	doi = {10.2307/2291067}
}

@article{prodmeasures,
	author = {R.-D. Reiss},
	title = {{Approximation of Product Measures with an Application to Order Statistics}},
	volume = {9},
	journal = {The Annals of Probability},
	number = {2},
	publisher = {Institute of Mathematical Statistics},
	pages = {335 -- 341},
	year = {1981},
	doi = {10.1214/aop/1176994477}
}

@Inbook{discretemc,
	author={Laurent Saloff-Coste},
	title={Lectures on finite Markov chains},
	bookTitle={Lectures on Probability Theory and Statistics: Ecole d'Et{\'e} de Probabilit{\'e}s de Saint-Flour XXVI-1996},
	year={1997},
	publisher={Springer Berlin Heidelberg},
	address={Berlin, Heidelberg},
	pages={301--413},
	doi={10.1007/BFb0092621}
}

@article{honestexp,
	author = {James P. Hobert and Galin L. Jones},
	title = {{Honest Exploration of Intractable Probability Distributions via Markov Chain Monte Carlo}},
	volume = {16},
	journal = {Statistical Science},
	number = {4},
	publisher = {Institute of Mathematical Statistics},
	pages = {312 -- 334},
	year = {2001},
	doi = {10.1214/ss/1015346317}
}

@Inbook{intromcmc,
	author={Charles J. Geyer},
	title={Introduction to Markov Chain Monte Carlo},
	bookTitle={Handbook of Markov Chain Monte Carlo},
	year={2011},
	publisher={Chapman and Hall/CRC},
	address={New York},
	pages={1 -- 46},
	doi={10.1201/b10905}
}

@article{geomgibbs,
	author = {Aixin Tan and Galin L. Jones and James P. Hobert},
	title = {On the Geometric Ergodicity of Two-Variable Gibbs Samplers},
	volume = {10},
	journal = {Institute of Mathematical Statistics Collections},
	pages = {25 -- 42},
	year = {2013},
	doi = {10.1214/12-IMSCOLL1002}
}

@article{gelmanrubin,
	author = {A. Gelman and D.B. Rubin},
	title = {Inference from Iterative Simulation using Multiple Sequences},
	volume = {7},
	number = {4},
	journal = {Statistical Science},
	pages = {457 -- 472},
	year = {1992},
	doi = {10.1214/ss/1177011136}
}

@article{clt,
author = {Galin L. Jones},
title = {{On the Markov chain central limit theorem}},
volume = {1},
journal = {Probability Surveys},
number = {none},
publisher = {Institute of Mathematical Statistics and Bernoulli Society},
pages = {299 -- 320},
year = {2004},
doi = {10.1214/154957804100000051},
URL = {https://doi.org/10.1214/154957804100000051}
}

@article{comp,
author = {Martin Dyer and Leslie Ann Goldberg and Mark Jerrum and Russell Martin},
title = {{Markov chain comparison}},
volume = {3},
journal = {Probability Surveys},
number = {none},
publisher = {Institute of Mathematical Statistics and Bernoulli Society},
pages = {89 -- 111},
year = {2006},
doi = {10.1214/154957806000000041},
URL = {https://doi.org/10.1214/154957806000000041}
}

@article{Jin2020CentralLT,
  title={Central limit theorems for Markov chains based on their convergence rates in Wasserstein distance},
  author={Rui Jin and Aixin Tan},
  journal={arXiv: Statistics Theory},
  year={2020}
}

@article{loccont,
author = {David Steinsaltz},
title = {{Locally Contractive Iterated Function Systems}},
volume = {27},
journal = {The Annals of Probability},
number = {4},
publisher = {Institute of Mathematical Statistics},
pages = {1952 -- 1979},
year = {1999},
doi = {10.1214/aop/1022874823},
URL = {https://doi.org/10.1214/aop/1022874823}
}

@article{faith,
title = {Faithful Couplings of Markov Chains: Now Equals Forever},
journal = {Advances in Applied Mathematics},
volume = {18},
number = {3},
pages = {372-381},
year = {1997},
issn = {0196-8858},
doi = {https://doi.org/10.1006/aama.1996.0515},
url = {https://www.sciencedirect.com/science/article/pii/S0196885896905151},
author = {Jeffrey S. Rosenthal}
}

\section{Web Appendix}\label{sec:appendix}

\subsection{Propositions related to the properties of total variation distance}

\begin{proof}[Proof of Proposition \ref{prop:invertibletv}]\label{proof:invertibletv}
	Let $\mathcal{A}$ be the sigma field of $\mathcal{X}$ and $\borel$ be the sigma field of $\mathcal{Y}$. 
	
	First note that $f^{-1}(\mathcal{B})=\{f^{-1}(B): B\in \mathcal{B}\}=\mathcal{A}$:
	\begin{itemize}
		\item \textbf{$f^{-1}(\borel)\subset \mathcal{A}$:} For $B\in \borel$, $f^{-1}(B)\subset \mathcal{A}$ by measurability.
		\item \textbf{$\mathcal{A} \subset f^{-1}(\borel)$:} Let $A\in \mathcal{A}$. Then $f(A)\in \borel$ and $f^{-1}(f(A))\in f^{-1}(\borel)$ by definition. By invertibility, $f^{-1}(f(A))=A$ and so $A \in f^{-1}(\borel)$.
	\end{itemize}  
	
	The equality in equation \ref{eqn:invertibletv} can then be proven as follows,
	\begin{align*}
		\norm{\L(f(X))-\L(f(X'))} &=\sup_{B\in f(\borel)}\abs{P(f(X)\in B)-P(f(X')\in B)}\\
		&= \sup_{B\in f(\borel)}\abs{P(X\in f^{-1}(B))-P(X'\in f^{-1}(B))}\\
		&= \sup_{A\in \mathcal{A}}\abs{P(X\in A)-P(X'\in A)} &\text{Since $f^{-1}(\borel)=\mathcal{A}$}\\
		&=\norm{\L(X)-\L(X')}
	\end{align*}
\end{proof}

\begin{proof}[Proof of Proposition \ref{prop:tvexp}]\label{proof:tvexp}
	\begin{align*}
		\norm{\L(X)-\L(X')} &= \sup_{A\in \borel} \abs{P(X\in A)-P(X'\in A)}\\
		&= \sup_{A\in \borel}\abs{\int_{\mathcal{Y}} P(X\in A\mid y)-P(X'\in A\mid y)\mu(dy)}\\
		&\leq \sup_{A\in \borel}\int_{\mathcal{Y}} \abs{P(X\in A\mid y)-P(X'\in A\mid y)}\mu(dy) &\text{ by Jensen's inequality}\\
		&\leq \int_{\mathcal{Y}} \sup_{A\in \borel} \abs{P(X\in A\mid y)-P(X'\in A\mid y)}\mu(dy) \\
		&\leq E\left[ \norm{\L(X\mid Y)-\L(X'\mid Y)} \right]
	\end{align*}
\end{proof}

\begin{proof}[Proof of Proposition \ref{prop:indcoord}]\label{proof:indcoord}
	To prove this we use the concept of maximal coupling over the coordinates. By maximal coupling, for $i\in \{1,\ldots, d\}$ there exists random variables $X_{i,n}^M, X_{i,n}^{'M}$ such that $X_{i,n}\overset{d}{=}X_{i,n}^M$ and $X'_{i,n}\overset{d}{=}X_{i,n}^{'M}$ and 
	\begin{equation*}
		\norm{\L(X_{i,n})-\L(X'_{i,n})}=P(X_{i,n}^M\neq X_{i,n}^{'M})
	\end{equation*}
	(see Proposition 3g of \cite{robrose} or Section 2 of \cite{maximalcoupling}).
	
	Further, there exists a unique product measure such that for any $A_1, \ldots A_d \in \cal{B}$, $P(\cap_{i=1}^d [X_{i,n}^M\in A_i])=\prod_{i=1}^d P(X_{i,n}^M\in A_i)$ (Theorem 18.2 of \cite{billingsley}). For the unique product measure, the following equality holds,
	\begin{align*}
		P(\cap_{i=1}^d X_{i,n}^M \in A_i)=\prod_{i=1}^d P(X_{i,n}^M\in A_i)
		=\prod_{i=1}^d P(X_{i,n}\in A_i)
		=P(\cap_{i=1}^d X_{i,n} \in A_i)
	\end{align*}
	And so by uniqueness, for $A\in \cal{B}^{\text{d}}$, $P(X_n^M\in A)=P(X_n\in A)$. By definition this means that $\vecx_{n}\overset{d}{=} \vecx_n^M$, which implies that $(\vecx_n^M, \vecx_n^{'M})\in \mathcal{C}(\vecx_n,\vecx'_n)$, the set of all couplings of $\vecx_n,\vecx'_n$.
	
	We now use $\vecx_n^M, \vecx_n^{'M}$ to prove equation \ref{eqn:indcoord}.
	\begin{align*}
		\norm{\L(\vec{X}_n)-\L(\vec{X}'_n)}&= \inf_{\vec{Y},\vec{Y}'\in \mathcal{C}(\vecx_n,\vecx'_n)}P(\vec{Y}\neq \vec{Y}') & \text{by equation 2.4 of \cite{maximalcoupling}}\\
		&\leq P(\vecx_n^M\neq \vecx_n^{'M})\\
		&= P(\cup_{i=1}^d [X_{i,n}^M\neq X_{i,n}^{'M}])\\
		&\leq \sum_{i=1}^d P(X_{i,n}^M\neq X_{i,n}^{'M}) &\text{by subadditivity}\\
		&\leq d A r^n
	\end{align*}
\end{proof}

\subsection{Lemmas related to the Sideways Theorem}
The following are lemmas and corresponding proofs and corollaries related to the Sideways Theorem (\ref{thm:oneshotlipschitz}).

\subsubsection{Lemmas providing an upper bound on the integral difference between a function and a corresponding shift}
 The following lemmas are used in the proof of Lemma \ref{lem:coalescingcondsideways}.
\begin{lem}\label{lem:shiftedfuncinvert}
	For any invertible, continuous function $f:\mathbb{R}\to \mathbb{R}$ where the codomain is $f(\mathbb{R})=(a,b)$ and $\Delta>0$,
	$$\int_{\mathbb{R}}\abs{f(x+\Delta)-f(x)}dx=(b-a)\Delta$$
	\begin{proof}
		Since $f$ is invertible and continuous, it is strictly monotone (Lemma 3.8 if \cite{analbyhist}). Assume that $f$ is strictly increasing. The integral can be written as follows,
		\begin{align*}
			\int_{\mathbb{R}}\abs{f(x+\Delta)-f(x)}dx&= \int_{\mathbb{R}}f(x+\Delta)-f(x)dx\\
			&=  \int_{\mathbb{R}}\int_a^bI_{f(x+\Delta)<y<f(x)}dy dx \\
			&=\int_{\mathbb{R}}\int_a^bI_{f^{-1}(y)-\Delta<x<f^{-1}(y)}dy dx \\
			&=\int_a^b \int_{\mathbb{R}}I_{f^{-1}(y)-\Delta<x<f^{-1}(y)}dx dy &\text{by Fubini's Theorem} \\
			&=\int_a^b \Delta dy \\
			&= (b-a)\Delta
		\end{align*} 
		
		If $f$ is strictly decreasing apply the transform $h(x)=a+b-f(x)$. The function $h$ is a strictly increasing invertible function with codomain $(a,b)$ and so using the previous result for increasing functions,
		\begin{align*}
			\int_{\mathbb{R}}\abs{f(x+\Delta)-f(x)}dx= \int_{\mathbb{R}}\abs{h(x+\Delta)-h(x)}dx = (b-a)\Delta
		\end{align*} 
	\end{proof}
\end{lem}

\begin{lem}\label{lem:shiftedfunc2}
	Let $f:\mathbb{R}\to \mathbb{R}$ be a continuous function that is invertible over the set $(c,d)$ and is a constant function over $(c,d)^C$. Further suppose that the codomain is $f(\mathbb{R})=(a,b)$. Then for $\Delta>0$ we get that
	$$\int_{\mathbb{R}}\abs{f(x+\Delta)-f(x)}dx=(b-a)\Delta$$
	\begin{proof}
		Assume that $f$ is an increasing function and so $f(c)=a$, $f(d)=b$ and $\abs{f(x+\Delta)-f(x)}=f(x+\Delta)-f(x)$. 
		
		Let $0<\epsilon<(c-d)/2$ and define 
		\[
		g_{\epsilon}(x) =
		\begin{dcases*}
			(f(c+\epsilon)-a)(1-e^{x-c-\epsilon})+a & when $x\in (-\infty,c+\epsilon]$\\
			f(x) & when $x\in (c+\epsilon,d-\epsilon]$\\
			(f(d-\epsilon)-b)(1-e^{d-\epsilon-x})+b & when $x\in(d-\epsilon,\infty)$\\
		\end{dcases*}
		\]
		Note that $g_{\epsilon}(x)$ is continuous, invertible, an increasing function and the codomain is $(a,b)$. By Lemma \ref{lem:shiftedfuncinvert} for each $\epsilon>0$
		$$\int_{\mathbb{R}}g_{\epsilon}(x+\Delta)-g_{\epsilon}(x) dx = (b-a)\Delta$$
		Further, for all $x\in \mathbb{R}$, $\lim_{\epsilon\to 0}g_{\epsilon}(x+\Delta)-g_{\epsilon}(x)=f(x+\Delta)-f(x)$ and so $g_{\epsilon}(x+\Delta)-g_{\epsilon}(x)$ converges pointwise to $f(x+\Delta)-f(x)$. Next, for $0<\epsilon<(c-d)/2$, $\abs{g_{\epsilon}(x+\Delta)-g_{\epsilon}(x)}<2\abs{b}$ and so the function $g_{\epsilon}(x+\Delta)-g_{\epsilon}(x)$ is uniformly bounded. The above statements allow us to apply the dominated convergence Theorem (Theorem 16.5 of \cite{billingsley}) and so
		$$\int_{\mathbb{R}}f(x+\Delta)-f(x)dx = \lim_{\epsilon \to 0} \int_{\mathbb{R}}g_{\epsilon}(x+\Delta)-g_{\epsilon}(x)dx = (b-a)\Delta$$
		
		If $f$ is strictly decreasing apply the transform $h(x)=a+b-f(x)$. The function $h$ is a strictly increasing invertible function with codomain $(a,b)$ and so using the previous result for increasing functions,
		\begin{align*}
			\int_{\mathbb{R}}\abs{f(x+\Delta)-f(x)}dx= \int_{\mathbb{R}}\abs{h(x+\Delta)-h(x)}dx = (b-a)\Delta
		\end{align*}
	\end{proof}
\end{lem}

\begin{lem}\label{lem:shiftedfunc3}
	Let $f:\mathbb{R}\to \mathbb{R}$ be a continuous function with the following properties:
	\begin{itemize}
		\item the codomain is $(0,K)$
		\item $(m_1, m_2, \ldots, m_M)$ are the local maxima and minima points
		\item $\lim_{x\to \infty}f(x)=0$ and $\lim_{x\to -\infty}f(x)=0$
	\end{itemize}
	Further suppose that $\Delta< \max_{i=2,\ldots, M}\{m_i-m_{i-1}\}$. Then 
	$$\int_{\mathbb{R}}\abs{f(x-\Delta)-f(x)}dx \leq K(M+1)\Delta$$
	
	\begin{proof}
		Since $\Delta< \max_{i=2,\ldots, M}\{m_i-m_{i-1}\}$, we have that $m_1-\Delta<m_1<m_2-\Delta<\ldots<m_M$. Let $I_1,\ldots, I_M$ be the intersection points or the points where $f(I_i)=f(I_i-\Delta)$. 
		
		\paragraph{Show that $m_i-\Delta<I_i<m_i$:} Suppose that $m_i$ is a local maximum point. Let $g(x)=f(x+\Delta)$. Within the interval $(m_i-\Delta,m_i)$, $f'(x)>0$ and $g'(x)<0$ by assumption. This implies that $f(m_i-\Delta)<f(m_i)$ and $g(m_i-\Delta)>g(m_i)$ by the Mean Value Theorem. Further since $g(m_i-\Delta)=f(m_i)$ we have that $g(m_i-\Delta)>f(m_i-\Delta)$ and $g(m_i)<f(m_i)$. 
		
		Let $h(x)=g(x)-f(x)$. Then $h(m_i-\Delta)>0$ and $h(m_i)<0$ further $h$ is a strictly decreasing function over $(m_i-\Delta,m_i)$ since $g,-f$ are strictly decreasing functions over the same interval. So by the intermediate value theorem, there exists an $\xi \in (m_i-\Delta,m_i)$ such that $h(\xi)=0$ or $f(\xi)=g(\xi)=f(\xi+\Delta)$. Further by injectivity, $\xi$ is unique. Let $I_i=\xi$. A similar proof can be given for when $m_i$ is a local minimum.
		
		\paragraph{Show that $\int_{I_i}^{I_{i+1}}\abs{f(x+\Delta)-f(x)}dx\leq K\Delta$:}
		Note first that $m_i-\Delta<I_i<m_i<m_{i+1}-\Delta<I_{i+1}<m_{i+1}$ further define
		\[
		f_i(x) =
		\begin{dcases*}
			f(m_i) & when $x\in (-\infty,m_i]$\\
			f(x) & when $x\in (m_i,m_{i+1}]$\\
			f(m_{i+1}) & when $x\in(m_{i+1},\infty)$\\
		\end{dcases*}
		\]
		Note that over the interval $(m_i,m_{i+1}]$, the function $f$ is either a strictly increasing or a strictly decreasing function.
		\begin{align*}
			&\int_{I_i}^{I_{i+1}}\abs{f(x+\Delta)-f(x)}dx \\
			&=\int_{I_i}^{m_i}\abs{f(x+\Delta)-f(x)}dx + \int_{m_i}^{m_{i+1}-\Delta}\abs{f(x+\Delta)-f(x)}dx +
			\int_{m_{i+1}-\Delta}^{I_{i+1}}\abs{f(x+\Delta)-f(x)}dx\\
			&\leq \int_{I_i}^{m_i}\abs{f(x+\Delta)-f(m_i)}dx + \int_{m_i}^{m_{i+1}-\Delta}\abs{f(x+\Delta)-f(x)}dx +
			\int_{m_{i+1}-\Delta}^{I_{i+1}}\abs{f(m_{i+1})-f(x)}dx\\
			&= \int_{I_i}^{m_i}\abs{f_i(x+\Delta)-f_i(x)}dx + \int_{m_i}^{m_{i+1}-\Delta}\abs{f_i(x+\Delta)-f_i(x)}dx +
			\int_{m_{i+1}-\Delta}^{I_{i+1}}\abs{f_i(x+\Delta)-f_i(x)}dx\\
			&= \int_{I_i}^{I_{i+1}}\abs{f_i(x+\Delta)-f_i(x)}dx \\
			&\leq \int_{m_i-\Delta}^{m_{i+1}}\abs{f_i(x+\Delta)-f_i(x)}dx \\
			&= \int_{\mathbb{R}}\abs{f_i(x+\Delta)-f_i(x)}dx \\
			&= \abs{f(m_i)-f(m_{i+1})}\Delta \leq K\Delta
		\end{align*}
		The last equality is a result of Lemma \ref{lem:shiftedfunc2}.
		
		By similar reasoning it can be shown that 
		$$\int_{-\infty}^{I_{1}}\abs{f(x+\Delta)-f(x)}dx\leq K\Delta \hspace{2cm} \int_{I_{M}}^{\infty}\abs{f(x+\Delta)-f(x)}dx\leq K\Delta$$ Finally note that the intersection points partition $\mathbb{R}$ into $M+1$ subsets and so
		$$\int_{\mathbb{R}}\abs{f(x-\Delta)-f(x)}dx \leq K(M+1)\Delta$$
	\end{proof}

\subsubsection{Proof of Lemma \ref{lem:coalescingcondsideways}}
Lemma \ref{lem:coalescingcondsideways} represents the coalescing condition for the Sideways Theorem \ref{thm:oneshotlipschitz}.
	
\begin{proof}[Proof of Lemma \ref{lem:coalescingcondsideways}] \label{pf:coalescingcondsideways}
	Set $\theta_{1,n}=\theta'_{1,n}$. Define $$\Delta=g(\theta_{1,n},X_{n-1})-g(\theta_{1,n},X'_{n-1})$$ Let $f_{X_n},f_{X'_n}$ be the density functions for $X_{n},X_{n}'$, respectively and $f_{\theta_{2,n}}, f_{\theta_{2,n}+\Delta}$ be the density functions for $\theta_{2,n}, \theta_{2,n}+\Delta$. 
	
	Suppose that $\Delta,X_{n-1}, X'_{n-1} \in \mathbb{R}$ are known and so,
	\begin{align*}
		X_n=g(\theta_{1,n},X_{n-1})+\theta_{2,n} 
		&\implies \theta_{2,n}=X_n-g(\theta_{1,n},X_{n-1}) \\
		X'_n=g(\theta_{1,n},X'_{n-1})+\theta'_{2,n} 
		&\implies \theta'_{2,n}-\Delta=X'_n-g(\theta_{1,n},X_{n-1})
	\end{align*}
	
	We know that $\theta_{2,n}\overset{d}{=}\theta '_{2,n}$ and in general $\Delta$, $\theta_{1,n}$ are random variables, so
	\begin{align}\label{eqn:tvbounded}
		\norm{\L(X_n)-\L(X'_n)}
		&\leq E_{\theta_{1,n}, \Delta}\left[\norm{\L(X_n\mid \theta_{1,n}, \Delta)-\L(X'_n\mid \theta_{1,n}, \Delta)}\right] &\text{by Proposition \ref{prop:tvexp}}\\
		&= E_{\theta_{1,n}, \Delta}\left[\norm{\L(\theta_{2,n}\mid\theta_{1,n})-\L(\theta_{2,n}-\Delta\mid \theta_{1,n})}\right] & \text{by Proposition \ref{prop:invertibletv}}
	\end{align}
	By the assumptions in the theorem, the density of $\theta_{2,n}$ is continuous with $M$ extrema points and has a codomain that is in $(0,K)$. Let $(m_1, m_2,\ldots, m_M)$ be the local extrema points where $m_i<m_j$ if $i<j$ and $L\leq \max_{2\leq i\leq M}\{m_i-m_{i-1}\}$ be the maximum distance between two local extrema points. So, continuing from the inequality \ref{eqn:tvbounded} and by the definition of total variation, equation \ref{eqn:tv},
	\begin{align*}
		\norm{\L(X_n)-\L(X'_n)} &\leq E_{\theta_{1,n}}\left[E_{\Delta}\left[\frac{1}{2}\int_{\mathbb{R}}\abs{f_{\theta_{2,n}}(x\mid \theta_{1,n})-f_{\theta_{2,n}-\Delta}(x\mid \theta_{1,n})}dx\right]\right] \\
		&= E_{\theta_{1,n}}\left[E_{\Delta}\left[\frac{1}{2}\int_{\mathbb{R}}\abs{f_{\theta_{2,n}}(x\mid \theta_{1,n})-f_{\theta_{2,n}}(x+\Delta\mid \theta_{1,n})}dx\right]\right]\\
		&= E_{\theta_{1,n}}\left[E_{\Delta}\left[\frac{1}{2}\int_{\mathbb{R}}\abs{f_{\theta_{2,n}}(x\mid \theta_{1,n})-f_{\theta_{2,n}}(x+\Delta\mid \theta_{1,n})}dxI_{\Delta<L}\right]\right]+\\ &\hspace{1cm}E_{\theta_{1,n}}\left[E_{\Delta}\left[\frac{1}{2}\int_{\mathbb{R}}\abs{f_{\theta_{2,n}}(x\mid \theta_{1,n})-f_{\theta_{2,n}}(x+\Delta\mid \theta_{1,n})}dxI_{\Delta>L}\right]\right]\\
		&\leq E_{\theta_{1,n}}\left[E_{\Delta}\left[\frac{1}{2}\int_{\mathbb{R}}\abs{f_{\theta_{2,n}}(x\mid \theta_{1,n})-f_{\theta_{2,n}}(x+\Delta\mid \theta_{1,n})}dx \middle| \abs{\Delta}<L\right]\right]+ P_{\Delta}(\abs{\Delta} >L)\\
		&\leq \frac{1}{2}E_{\theta_{1,n}}\left[E_\Delta\left[K(M+1)\abs{\Delta} \right]\right] + P_\Delta(\abs{\Delta}>L)&\text{by Lemma \ref{lem:shiftedfunc3}}\\
		&\leq \frac{K(M+1)}{2}E_\Delta\left[\abs{\Delta} \right] +\frac{E_{\Delta}[\abs{\Delta}]}{L}
	\end{align*}
	The coalescing condition is thus satisfied as follows with $C=\frac{K(M+1)}{2} +\frac{I_{M>1}}{L}$,
	\begin{align*}
		\norm{\L(X_{n+1})-\L(X'_{n+1})} &\leq C E[\abs{g(\theta_{1,n},X_{n-1})-g(\theta_{1,n},X'_{n-1})}]\\
		&=CE[\abs{g(\theta_{1,n},X_{n-1}) + \theta_{2,n}-(g(\theta_{1,n},X'_{n-1}) +\theta_{2,n})}]\\
		&= C E[\abs{X_n-X'_n}]
	\end{align*}	
%
\end{proof}
\end{lem}

\subsection{Lemmas for random-functional autoregressive process examples}
\subsubsection{Proof of Lemma \ref{lem:nonlinar}}
\begin{proof}[Proof of Lemma \ref{lem:nonlinar}]\label{proof:nonlinar}
	First note that 
	\begin{align*}
		&E[\abs{X_{n+2}-X'_{n+2}} \mid X_n=x,X'_n=y] \\
		&= E\left[\bigg\lvert \left(\frac{1}{2}(x -\sin x)+Z_n\right)-g\left(\frac{1}{2}(y -\sin y)+Z_n\right)\bigg\rvert\right] \\
		&= E\left[\Abs{\frac{1}{2}\left(\frac{1}{2}(x -\sin x)+Z_n -\sin\left(\frac{1}{2}(x -\sin x)+Z_n\right)\right)-\frac{1}{2}\left(\frac{1}{2}(y -\sin y)+Z_n -\sin\left(\frac{1}{2}(y -\sin y)+Z_n\right)\right)}\right] \\
		&= \frac{1}{2}E\left[\Abs{\frac{1}{2}(x-y +\sin y -\sin x) + \sin\left(\frac{1}{2}(y -\sin y)+Z_n\right) -\sin\left(\frac{1}{2}(x -\sin x)+Z_n\right)}\right]\\
		&= \frac{1}{2}E\left[\abs{g(x,y) + G(x,y)}\right]
	\end{align*}
	Where $g(x,y)=\frac{1}{2}(x-y +\sin y -\sin x)$ and $G(x,y)=\sin\left(\frac{1}{2}(y -\sin y)+Z_n\right) -\sin\left(\frac{1}{2}(x -\sin x)+Z_n\right)$. By trigonometric identities \footnote{The trigonometric identities used are $2\cos\mu \sin \upsilon = \sin(\mu+\upsilon)-\sin(\mu-\upsilon)$ and $\cos(\mu + \upsilon)= \cos\mu \cos \upsilon +\sin \mu \sin \upsilon$ where $\mu,\upsilon\in \mathbb{R}$}, for $k(x,y)=\frac{x+y-\sin y - \sin x}{4}$ and $h(x,y)=\frac{y-x+\sin x - \sin y}{4}$. 
	\begin{align*}
		G(x,y) &= 2\cos\left(\frac{x+y-\sin y - \sin x}{4} + Z_n\right)\sin\left(\frac{y-x+\sin x - \sin y}{4}\right) \\
		&= 2\cos\left(k(x,y) + Z_n\right)\sin h(x,y) \\
		&= 2\sin h(x,y) \left(\cos Z_n \cos k(x,y) + \sin Z_n \sin k(x,y)\right)
	\end{align*}
	And so,
	\begin{align*}
		& E[\abs{X_{n+2}-X'_{n+2}} \mid X_n=x,X'_n=y]\\
		& = \frac{1}{2}E\left[\Abs{g(x,y) + 2\sin h(x,y) \left(\cos Z_n \cos k(x,y) + \sin Z_n \sin k(x,y)\right)}\right] \\
		&\leq \frac{1}{2}\sqrt{E\left[\left(g(x,y) + 2\sin h(x,y) \left(\cos Z_n \cos k(x,y) + \sin Z_n \sin k(x,y)\right)\right)^2\right]} \\
		&= \frac{1}{2}\sqrt{g(x,y)^2 + 4e^{-1/2}g(x,y) \sin h(x,y)\cos k(x,y) + 4\sin^2 h(x,y) E[\left(\cos Z_n \cos k(x,y) + \sin Z_n \sin k(x,y)\right)^2]} \\
		&= \frac{1}{2}\sqrt{g(x,y)^2 + 4e^{-1/2}g(x,y) \sin h(x,y)\cos k(x,y) + 2\sin^2 h(x,y) (1+e^{-2}(\cos^2k(x,y) - \sin^2 k(x,y)))}\\
		&= \frac{1}{2}\sqrt{4h(x,y)^2 - 8e^{-1/2}h(x,y) \sin h(x,y)\cos k(x,y) + 2\sin^2 h(x,y) (1+e^{-2}(\cos^2k(x,y) - \sin^2 k(x,y)))}
	\end{align*}

\end{proof}

\subsubsection{Proof of lemmas used in Theorem \ref{thm:bayesianreggibbs}}

To prove the first part of this theorem, we apply the de-initialization technique which shows how the convergence rate of a Markov chain can be bounded above by the convergence rate of a more simpler Markov chain that includes sufficient information on the Markov chain of interest. The concept of de-initialization and a proposition that bounds total variation is provided below.

\begin{defn}[De-initialisation]
	Let $\{X_n\}_{n\geq 1}$ be a Markov chain. A Markov chain $\{Y_n\}_{n\geq 1}$ is a de-initialization of $\{X_n\}_{n\geq 1}$ if for each $n\geq 1$ $$\L(X_n \mid X_0,Y_n)=\L(X_n\mid Y_n)$$
\end{defn}

\begin{prop}[Theorem 1 of \cite{deinit}]\label{prop:deinit}
	Let $\{Y_n\}_{n\geq 1}$ be a de-initialization of $\{X_n\}_{n\geq 1}$ then for any two initial distributions $X_0\sim \mu$ and $X'_0\sim \mu'$,
	$$\norm{\L(X_n)-\L(X'_n)}\leq \norm{\L(Y_n)-\L(Y'_n)}$$
\end{prop}

\begin{proof}[Proof of Lemma \ref{lem:bayesregdeinit}]\label{proof:bayesregdeinit}
	Note that $\beta_{n}= \tilde{\beta} + \sigma_{n-1} Z_n, Z_n\sim N_p(0, A^{-1})$ can be written as a random function of $\sigma^2_n$. Substituting $\beta_n$, $\sigma^2_{n}$ can then be written as a random function of its previous value for independent $Z^2_{n}\sim \chi^2(p)$ and $G_n \sim \Gamma(\frac{k+p}{2},1)$,
	
	$$\sigma^2_{n}=\frac{Z^2_{n}}{C}\frac{C}{2G_n}\sigma^2_{n-1}+\frac{C}{2G_n}$$
	
	Let $X_n=\frac{Z^2_{n}}{C}$, $Y_n=\frac{C}{2G_n}$. We can rewrite $\sigma^2_{n}=X_nY_n\sigma^2_{n-1}+Y_n$
	where $X_n\sim \Gamma\left(\frac{p}{2}, \frac{C}{2}\right)$ and $Y_n\sim \Gamma^{-1}\left(\frac{k+p}{2}, \frac{C}{2}\right)$.
	Using the notation from the Sideways Theorem \ref{thm:oneshotlipschitz} $\theta_{1,n}=X_nY_n$ and $\theta_{2,n}=Y_n$.
	
	Since $\beta_{n}$ can be written as a random function of $\sigma^2_n$,
	$$\L(\beta_n,\sigma^2_n\mid \beta_0,\sigma^2_0,\sigma^2_n)=\L(\beta_n,\sigma^2_n\mid \sigma^2_n)$$ 
	and so $\sigma^2_n$ is a de-initialization of $(\beta_n,\sigma^2_n)$. By Proposition \ref{prop:deinit}, $$\norm{\L(\beta_n,\sigma^2_n)-\L(\beta'_n, \sigma^{'2}_n)}\leq \norm{\L(\sigma^2_n)-\L(\sigma^{'2}_n)}$$
	We are thus interested in evaluating the convergence rate of $\sigma^2_{n}$ to bound the convergence rate of $(\beta_{n},\sigma^2_{n})$.
	
	To interpret this in another way, if $\sigma^2_n$ couples then the distribution of $\beta_n$ is the same for both iterations, so it is automatically coupled. An alternative proof can be made using the results from \cite{liu}.
\end{proof}

%

\begin{proof}[Proof of Lemma \ref{lem:bayesregcontr}]\label{proof:bayesregcontr}
	By Lemma \ref{lem:bayesregdeinit}, $\theta_{1,n}=X_n Y_n$ and so,
	\begin{align*}
		K=E[\abs{\theta_{1,n}}] = E[X_n Y_n]= E[X_n]E[Y_n]= \frac{p}{C}\frac{C}{k+p-2}= \frac{p}{k+p-2}
	\end{align*}
\end{proof}

\begin{proof}[Proof of Lemma \ref{lem:bayesregconddens}]\label{proof:bayesregconddens}
	\textbf{Calculate the conditional density $\theta_{2,n}\mid\theta_{1,n}$}
	We remove the subscript $n$ on the random variables. Let $X,Y$ be as described in Lemma \ref{lem:bayesregdeinit}. Since the random variables are independent, the joint density is the product of the densities. 
	\begin{equation}
		f_{X,Y}(x,y)=\frac{\bx}{\Gamma(\ax)}x^{\ax-1}e^{x\bx} \frac{\by}{\Gamma(\ay)}y^{-\ay-1}e^{-\frac{\by}{y}}
	\end{equation}
	Then $(\theta_1, \theta_2)=(XY,Y)$ is a transformation with the Jacobian $\abs{J}=\theta_2^{-1}$ and the density written as follows,
	\begin{align*}
		f_{\theta_1,\theta_2}(\theta_1,\theta_2)&= f_{X,Y}\left(\frac{\theta_1}{\theta_2},\theta_2\right)\theta_2^{-1}\\
		&=\frac{\bx}{\Gamma(\ax)}\left(\frac{\theta_1}{\theta_2}\right)^{\ax-1}e^{-\frac{\theta_1}{\theta_2}\bx} \frac{\by}{\Gamma(\ay)}\theta_2^{-\ay-1}e^{-\frac{\by}{\theta_2}}\theta_2^{-1}
	\end{align*}
	Next $f_{\theta_2\mid \theta_1}(\theta_2\mid\theta_1)$ is proportional to $f_{\theta_1,\theta_2}(\theta_1,\theta_2)$ and so we can derive the conditional density of $\theta_2$ as follows,
	\begin{align}\label{eqn:theta2density}
		f_{\theta_2\mid \theta_1}(\theta_2\mid \theta_1) &\propto f_{\theta_1,\theta_2}(\theta_1,\theta_2)\\
		&\propto \theta_2^{1-\ax}e^{-\frac{1}{\theta_2}\theta_1\bx} \theta_2^{-\ay-1}e^{-\frac{1}{\theta_2}\by}\theta_2^{-1} \\
		&= \theta_2^{-(\ax+\ay)-1}e^{-\frac{1}{\theta_2}(\theta_1+1)\bx}
	\end{align}
	This is proportional to an inverse gamma distribution and so, 
	$\theta_2\mid \theta_1 \sim \Gamma^{-1}\left(\frac{k+2p}{2}, (\theta_1+1)C/2\right)$. Since the conditional density is an inverse gamma distribution, the number of modes is $M=1$ and the density function is continuous.
	
	\textbf{Calculate the maximum value of $f_{\theta_2\mid \theta_1}(\theta_2\mid \theta_1)$ :} Figure \ref{fig:inversegammadensity} shows how the maximum value of the density increases as the shape, $(\theta_1+1)C/2$ decreases when the rate, $\frac{k+2p}{2}$ is fixed. It can also be shown from equation \ref{eqn:theta2density} that the density function of $f_{\theta_2\mid \theta_1}(\theta_2\mid \theta_1)$ is maximized when $\theta_1=0$ since the normalizing constant will be the largest. This means that $f_{\theta_2\mid \theta_1}(\theta_2\mid \theta_1)$ reaches its maximum height when $\theta_1=0$ and so we find the value of $f_{\theta_2\mid \theta_1}(\theta_2\mid \theta_1)$ evaluated at $\theta_2= \frac{C}{k+2p+2}$, the mode (Section 5.3 of \cite{mode}).
	\begin{align*}
		K &= f_{\theta_2\mid \theta_1}\left(\frac{C}{k+2p+2}\mid\theta_1=0\right)\\
		&=\frac{(C/2)^{\frac{k+2p}{2}}}{\Gamma(\frac{k+2p}{2})}y^{-\frac{k+2p}{2}-1}e^{-\frac{C/2}{y}}\mid_{y=\frac{C}{k+2p+2}}\\
		&=\frac{(C/2)^{\frac{k+2p}{2}}}{\Gamma(\frac{k+2p}{2})}\left(\frac{C}{k+2p+2}\right)^{-\frac{k+2p}{2}-1}e^{-\frac{k+2p+2}{2}}\\
		&=\frac{(C/2)^{\frac{k+2p}{2}}}{\Gamma(\frac{k+2p}{2})}\left(\frac{k+2p+2}{C}\right)^{\frac{k+2p}{2}+1}e^{-\frac{k+2p+2}{2}}
	\end{align*}
	
	And so,
	\begin{align}
		K &= \frac{(C/2)^{\frac{k+2p}{2}}}{\Gamma(\frac{k+2p}{2})}\left(\frac{k+2p+2}{C}\right)^{\frac{k+2p}{2}+1}e^{-\frac{k+2p+2}{2}}
	\end{align}
	
	\begin{figure}
		\centering
		\includegraphics[width=0.7\linewidth]{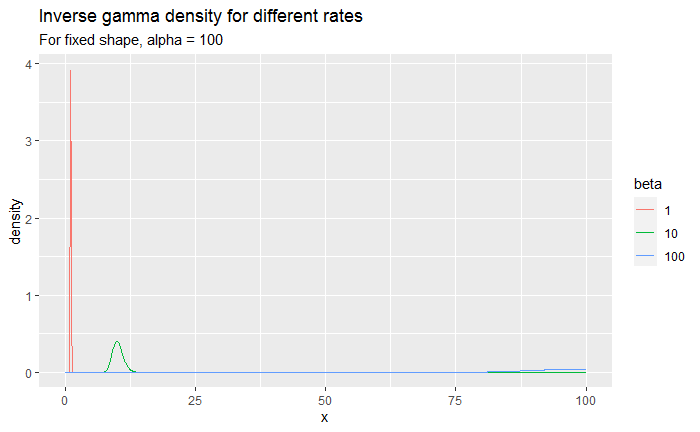}
		\caption{Inverse gamma density when $\alpha=100$ and $\beta=1,10,100$}
		\label{fig:inversegammadensity}
	\end{figure}
\end{proof}

\subsubsection{Proof of lemmas used in Theorem \ref{thm:anothergibbssampler}}

\begin{proof}[Proof of Lemma \ref{lem:locmodeldeinit}]\label{proof:locmodeldeinit}
	The iteration $\tau^{-1}_{n+1}$ can be written as a function of its previous value, $\tau^{-1}_{n}$ since $\mu_{n+1} = \bar{y} + Z_{n+1}/\sqrt{J \tau_{n}}$. 
	\begin{equation}\label{var}
		\tin_{n+1} = \frac{Z^2_{n+1}}{S}\frac{S}{2G_{n+1}}\tin_{n} + \frac{S}{2G_{n+1}}
	\end{equation}
	Next we can rewrite, $\tinv_{n}=X_nY_n \tinv_{n-1}+Y_n$ where $X_n=\frac{Z^2_{t+1}}{S}\sim \Gamma\left(\frac{1}{2}, \frac{S}{2}\right)$ and $Y_n=\frac{S}{2G_{t+1}}\sim \Gamma^{-1}\left(\frac{J+2}{2}, \frac{S}{2}\right)$.

	Since $(\mu_{n},\tau^{-1}_{n})$ can be written as a random function of $\tau^{-1}_n$,
	$$\L(\mu_n,\tau^{-1}_n\mid \mu_0,\tau^{-1}_0,\tau^{-1}_n)=\L(\mu_n,\tau^{-1}_n\mid \tau^{-1}_n)$$ 
	and $\tau^{-1}_n$ is a de-initialization of $(\mu_n,\tau^{-1}_n)$. Further, by Proposition \ref{prop:deinit}, $$\norm{\L(\mu_n,\tinv_n)-\L(\mu'_n, \tau^{'-1}_n)}\leq \norm{\L(\tinv_n)-\L(\tau^{'-1}_n)}$$
	
	To interpret this in another way, if $\tau_n$ couples then the distribution of $\mu_n$ is the same for both iterations, so it is automatically coupled. An alternative proof can be made using the results from \cite{liu}.
\end{proof}

%

\begin{proof}[Proof of Lemma \ref{lem:locmodelcontr}]\label{proof:locmodelcontr}
	By Lemma \ref{lem:locmodeldeinit}, $\theta_{1,n}=X_nY_n$ and so by Corollary \ref{cor:oneshotlipschitz}
	\begin{align*}
		D= E[\abs{\theta_{1,n}}] = E[X_nY_n]= E[X_n]E[Y_n]= \frac{1}{S}\frac{S}{J}= \frac{1}{J}
	\end{align*}
\end{proof}

\begin{proof}[Proof of Lemma \ref{lem:locmodelconddens}]\label{proof:locmodelconddens}
	To find $M,K$ and show that the conditional density is continuous, we (a) show that $\theta_2\mid \theta_1 \sim \Gamma^{-1}\left(\frac{J-1}{2}, (\theta_1+1)S/2\right)$, which directly implies that the conditional distribution is continuous and $M=1$ and we (b) we find the value of $K$.
	
	\textbf{(a) Calculate the conditional density $\theta_{2,n}\mid \theta_{1,n}$}
	For simplicity, we remove the subscript $n$ on the random variables. Let $X,Y$ be as described in Lemma \ref{lem:locmodeldeinit}. Since the random variables are independent, the joint density is the product of the densities. 
	\begin{equation}
		f_{X,Y}(x,y)=\frac{\bbx}{\Gamma(\aax)}x^{\aax-1}e^{x\bbx} \frac{\bby}{\Gamma(\aay)}y^{-\aay-1}e^{-\frac{\bby}{y}}
	\end{equation}
	Then $(\theta_1, \theta_2)=(XY,Y)$ is a transformation with the Jacobian $\abs{J}=\theta_2^{-1}$ and the density written as follows,
	\begin{align*}
		f_{\theta_1,\theta_2}(\theta_1,\theta_2)&= f_{X,Y}\left(\frac{\theta_1}{\theta_2},\theta_2\right)\theta_2^{-1}\\
		&=\frac{\bbx}{\Gamma(\aax)}\left(\frac{\theta_1}{\theta_2}\right)^{\aax-1}e^{-\frac{\theta_1}{\theta_2}\bbx} \frac{\bby}{\Gamma(\aay)}\theta_2^{-\aay-1}e^{-\frac{\bby}{\theta_2}}\theta_2^{-1}
	\end{align*}
	Next $f_{\theta_2\mid \theta_1}(\theta_2\mid \theta_1)$ is proportional to $f_{\theta_1,\theta_2}(\theta_1,\theta_2)$ and so we can derive the conditional density of $\theta_2$ as follows,
	\begin{align}\label{eqn:theta2densitycont}
		f_{\theta_2\mid \theta_1}(\theta_2\mid \theta_1) &\propto f_{\theta_1,\theta_2}(\theta_1,\theta_2)\\
		&\propto \theta_2^{1-\aax}e^{-\frac{1}{\theta_2}\theta_1\bbx} \theta_2^{-\aay-1}e^{-\frac{1}{\theta_2}\bby}\theta_2^{-1} \\
		&= \theta_2^{-(\aax+\aay)-1}e^{-\frac{1}{\theta_2}(\theta_1+1)\bbx} \\
		&= \theta_2^{-(J-1)/2-1}e^{-\frac{1}{\theta_2}(\theta_1+1)\bbx}
	\end{align}
	This is proportional to an inverse gamma distribution and so, 
	$\theta_2\mid \theta_1 \sim \Gamma^{-1}\left(\frac{J-1}{2}, (\theta_1+1)S/2\right)$. We know that the inverse gamma distribution is continuous and unimodal, so $M=1$.
	
	\textbf{(b) Calculate the maximum value of $f_{\theta_2\mid \theta_1}(\theta_2\mid \theta_1)$ :} Similar to figure \ref{fig:inversegammadensity} of Example \ref{ex:bayesianreggibbs}, $f_{\theta_2\mid\theta_1}(\theta_2\mid \theta_1)$ is maximized when $\theta_1=0$ since the normalizing constant will be the largest. So the largest value of $f_{\theta_2\mid\theta_1}(\theta_2\mid \theta_1)$ will occur when $\theta_1=0$. To find the maximum conditional distribution, we find the value of $f_{\theta_2\mid \theta_1}(\theta_2\mid \theta_1=0)$ evaluated at $\theta_2= \frac{S}{J+1}$, the mode (see Section 5.3 of \cite{mode}).
	\begin{align*}
		K &= f_{\theta_2\mid \theta_1}\left(\frac{S}{J+1}\mid\theta_1=0\right)\\
		&=\frac{(S/2)^{\frac{J-1}{2}}}{\Gamma(\frac{J-1}{2})}y^{-\frac{J-1}{2}-1}e^{-\frac{S/2}{y}}\mid_{y=\frac{S}{J+1}}\\
		&=\frac{(S/2)^{\frac{J-1}{2}}}{\Gamma(\frac{J-1}{2})}\left(\frac{S}{J+1}\right)^{-\frac{J-3}{2}}e^{-\frac{J+1}{2}}
	\end{align*}
	And so,
	\begin{align}\label{eqn:K_agibbssamplerex}
		K &= \frac{(S/2)^{\frac{J-1}{2}}}{\Gamma(\frac{J-1}{2})}\left(\frac{S}{J+1}\right)^{-\frac{J-3}{2}}e^{-\frac{J+1}{2}}
	\end{align}
\end{proof}

\begin{proof}[Proof of lemma \ref{lem:expstat}]\label{proof:expstat}
		By the property of stationary distribution, if $\sigma^2_{n-1}\sim \pi$ then $\sigma^2_{n}\sim \pi$ and so the lemma follows from the following.
		$$E_{\sigma^2_n \sim \pi}[V(\sigma^2_n)] =E_{\sigma^2_{n-1}\sim \pi}[E[V(\sigma^2_n)\mid \sigma^2_{n-1}]]\leq E_{\sigma^2_{n-1} \sim \pi}[\lambda V(\sigma^2_{n-1}) +b] = \lambda E_{\sigma^2_{n} \sim \pi}[ V(\sigma^2_{n})]+b$$
\end{proof}

\begin{proof} [Proof of \ref{lem:locmodeldriftval}]\label{proof:locmodeldriftval}
	\begin{align*}
		E[V(\sigma^2_n)\mid \sigma^2_{n-1}] &= E[(\sigma^2_n-h)^2\mid \sigma^2_{n-1}]\\
		&=E[(\sigma^2_n)^2-2h \sigma^2_n + h^2\mid\sigma^2_{n-1}]\\
		&=E[(X_nY_n \sigma^2_{n-1} + Y_n)^2-2h (X_nY_n \sigma^2_{n-1} + Y_n) + h^2\mid\sigma^2_{n-1}]\\
		&=E[Y_n^2](E[X_n^2](\sigma^2_{n-1})^2 + 2E[X_n]\sigma^2_{n-1} + 1)-2h (E[X_n] E[Y_n] \sigma^2_{n-1} + E[Y_n]) + h^2\\
		&=E[Y_n^2]E[X_n^2](\sigma^2_{n-1})^2 + 2E[X_n]E[Y_n^2]\sigma^2_{n-1} + E[Y_n^2]-2hE[X_n] E[Y_n] \sigma^2_{n-1} -2hE[Y_n] + h^2\\
		&=E[Y_n^2]E[X_n^2](\sigma^2_{n-1})^2 + 2E[X_n](E[Y_n^2]-h E[Y_n])\sigma^2_{n-1} + E[Y_n^2]-2hE[Y_n] + h^2\\
		&=0.6583702(\sigma^2_{n-1})^2 + 0.6911206\sigma^2_{n-1} + 107.3691\\
		&=\lambda(\sigma^2_{n-1})^2 + 2\lambda h\sigma^2_{n-1} + \lambda h^2 +b\\
		&=\lambda(\sigma^2_{n-1}+h)^2 +b
	\end{align*}
\end{proof}

\subsubsection{Proof of Theorem \ref{thm:arnormd}}

\begin{proof}[Proof of Theorem \ref{thm:arnormd}]\label{proof:arnormd}
	This example uses a modified version of the Sideways Theorem \ref{thm:oneshotlipschitz} to find an upper bound on the convergence rate. We will also use Proposition \ref{prop:invertibletv}, which states that the total variation between two random variables is equal to the total variation of any invertible transformation of the same two random variables.
	
	Let $\vecx _n, \vecx '_n \in \mathbb{R}^2$ be two copies of the autoregressive normal process as defined in Example \ref{ex:arnormdep}. Then for $\vec{Z}_n\sim N(\vec{0},I_d)$,
	$$\vecx_n=A\vecx_{n-1}+\Sigma_d\vec{Z}_n \hspace{2cm} \vecx '_n=A\vecx '_{n-1}+\Sigma_d\vec{Z}'_n$$
	We apply the one-shot coupling method to bound the total variation distance. For $n<N$ set $\vec{Z}_n=\vec{Z} '_n$.
	
	Suppose $X_0, X'_0$ are known and define 
	$$\Delta = \norm{\Sigma^{-1}_d A^n (\vecx_{0}-\vecx '_{0})}_2$$
	Decompose $A=P D P^{-1}$ with $D$ as the corresponding diagonal matrix, $\lambda_i$ is the $i$th eigenvalue of $A$ and $\norm{\cdot}_2$ denotes the Frobenius norm. Then $\Delta$ is bounded above as follows,
	\begin{align*}
		\Delta &= \norm{\Sigma^{-1}_d A^n (\vecx_{0}-\vecx '_{0})}_2\\
		&=\norm{\Sigma^{-1}_d P D^n P^{-1} (\vecx_{0}-\vecx '_{0})}_2\\
		&\leq \norm{\Sigma^{-1}_d}_2 \cdot \norm{P}_2 \norm{D^n}_2 \norm{P^{-1}}_2 \norm{\vecx_{0}-\vecx '_{0}}_2 &\text{by Lemma 1.2.7 of \cite{frobnorm}}\\
		&\leq \norm{\Sigma^{-1}_d}_2 \cdot \norm{P}_2 \norm{P^{-1}}_2 \norm{\vecx_{0}-\vecx '_{0}}_2 \sqrt{\sum_{i=1}^d \abs{\lambda_i}^{2n}}\\
		&\leq \norm{\Sigma^{-1}_d}_2 \cdot \norm{P}_2 \norm{P^{-1}}_2 \norm{\vecx_{0}-\vecx '_{0}}_2 \sqrt{d} \max_{1\leq i\leq d}\abs{\lambda_i}^n	
	\end{align*}
	
	For now assume that $X_0, X'_0$ are known and note that $\Sigma^{-1}_d$ is an invertible transform. We bound the total variation distance as follows by applying two invertible transforms on the Markov chain and using the fact that $\vec{Z}_{m}=\vec{Z}'_m, m < N$.
	\begin{align*}
		&\norm{\L(\vecx_N)-\L(\vecx '_N)} \\
		&\leq E_{\{\vec{Z}_m\}_{m<N}}\left[\norm{\L(\vecx_N)-\L(\vecx '_N)}\right] &\text{by prop. \ref{prop:tvexp}}\\
		&= E_{\{\vec{Z}_m\}_{m<N}}\left[\norm{\L(\Sigma^{-1}_d\vecx_N)-\L(\Sigma^{-1}_d\vecx '_N)}\right] &\text{by prop. \ref{prop:invertibletv}}\\
		&= E_{\{\vec{Z}_m\}_{m<N}}\left[\norm{\L(\Sigma^{-1}_d A \vecx_{N-1} +\vec{Z}_N)-\L(\Sigma^{-1}_dA \vecx '_{N-1} +\vec{Z} '_N)} \right]\\
		&= E_{\{\vec{Z}_m\}_{m<N}}\left[\norm{\L(\Sigma^{-1}_d (A^{N} \vec{X}_0 + \sum_{m=1}^{N-1} A^{N-m}\vec{Z}_m) +\vec{Z}_N)-\L(\Sigma^{-1}_d  (A^{N} \vec{X}'_0 + \sum_{m=1}^{N-1}A^{N-m}\vec{Z}_m) +\vec{Z} '_N)} \right]\\
		&= E_{\{\vec{Z}_m\}_{m<N}}\left[\norm{\L(\Sigma^{-1}_d A^N \vec{X}_0 +\vec{Z}_N)-\L(\Sigma^{-1}_d  A^N \vec{X}'_0 +\vec{Z} '_N)} \right] & \text{by prop. \ref{prop:invertibletv}}\\
		&=E_{\{\vec{Z}_m\}_{m<N}}\left[\norm{\L(\vec{Z}_N +\Sigma^{-1}_dA^N (\vecx_{0} -\vecx '_{0}))-\L(\vec{Z} '_N)}\right]\\
		&= \norm{\L(\vec{Z}_N +\Sigma^{-1}_dA^N (\vecx_{0} -\vecx '_{0}))-\L(\vec{Z} '_N)}
	\end{align*}
	There exists a rotation matrix $R\in \mathbb{R}^{d\times d}$ such that 
	$$R[\Sigma^{-1}_dA (\vecx_n -\vecx '_n)]=(\norm{\Sigma^{-1}_dA (\vecx_n -\vecx '_n)}_2,0,\ldots 0) = (\Delta,0,\ldots 0)$$ 
	\cite{frobnorm}. By properties of rotation, $R$ is orthogonal, so $R^T =R^{-1}$ and $RZ_n \sim N(0,RI_d R^T)=N(0,I_d)\sim Z_n$. In other words, $RZ_n \overset{d}{=} Z_n \overset{d}{=} Z'_n$. Thus, continuing the above equality,
	\begin{align*}
		\norm{\L(\vecx_n)-\L(\vecx '_n)} 	&\leq \norm{\L(\vec{Z}_n +\Sigma^{-1}_dA^n (\vecx_{0} -\vecx '_{0}))-\L(\vec{Z} '_n)}\\
		&= \norm{\L(R[\vec{Z}_n +\Sigma^{-1}_dA (\vecx_n -\vecx '_n)])-\L(R\vec{Z} '_n)}  & \text{by prop. \ref{prop:invertibletv}}\\
		&=\norm{\L(\vec{Z}_n +(\Delta,0,\ldots 0))-\L(\vec{Z} _n)}
	\end{align*}
	
	Next, suppose that $X_0, X'_0$ are unknown. Then, the inequality stated in equation \ref{eqn:arnormd} is shown as follows,
	\begin{align*}
		\norm{\L(\vecx_n)-\L(\vecx '_n)} &\leq E_{\Delta}[\norm{\L(\vec{Z}_n +(\Delta,0,\ldots 0))-\L(\vec{Z} _n)}] &\text{by prop \ref{prop:tvexp}}\\
		&= E_{\Delta}[\frac{1}{2}\int_{\mathbb{R}^d} \Abs{\frac{1}{(2\pi)^{d/2}}e^{-y_1^2/2 -\sum_{i=2}^d y_i^2/2}-\frac{1}{(2\pi)^{d/2}}e^{-(y_1-\Delta)^2/2 -\sum_{i=2}^d y_i^2/2}}d\vec{y} ]\\
		&= E_{\Delta}[\frac{1}{2}\int_{\mathbb{R}} \abs{\frac{1}{\sqrt{2\pi}}e^{-y_1^2/2 }-\frac{1}{\sqrt{2\pi}}e^{-(y_1-\Delta)^2/2}d}\vec{y} ]\\
		&= E_{\Delta}[\norm{\L(Z_{1,n}+\Delta)-\L(Z_{1,n})}]\\
		&\leq \frac{1}{\sqrt{2\pi}} E[\Delta] &\text{by Lemma \ref{lem:shiftedfunc3}}\\
		&\leq \sqrt{\frac{d}{2\pi}} \norm{ \Sigma^{-1}_d}_2 \cdot \norm{ P}_2 \norm{P^{-1}}_2 E[\norm{\vecx_{0}-\vecx '_{0}}_2] \max_{1\leq i\leq d}\abs{\lambda_i}^n
	\end{align*}
	
\end{proof}

\subsection{Lemmas for ARCH process examples}
\subsubsection{Proof of lemmas used in Theorem \ref{thm:lineararch}}

\begin{proof}[Proof of Lemma \ref{lem:larchcontr}]\label{proof:larchcontr}
	Let $\{X_n\}_{n\geq 1}\in \mathbb{R}$ and $\{X'_n\}_{n\geq 1}\in \mathbb{R}$ be two copies of the LARCH process.
	For fixed $n\geq 1$, let $Z_n=Z'_n$ and so,
	\begin{align*}
		E[\abs{X_n-X'_n}]&= E[\abs{(\beta_0+\beta_1 X_{n-1})Z_n-(\beta_0+\beta_1 X'_{n-1})Z_n}]\\
		&\leq \beta_1 E[\abs{Z_n}] E[\abs{X_{n-1}-X'_{n-1}}]
	\end{align*}
	Since $Z_n\overset{d}{=} Z_0>0$ a.s., the geometric convergence rate is $D=\beta_1 E[Z_0]$.
\end{proof}

\begin{proof}[Proof of Lemma \ref{lem:larchcoalesc}]\label{proof:larchcoalesc}
	For a fixed $n\geq 0$, suppose that $Z_{n+1}, Z'_{n+1}$ are independent. By Proposition \ref{prop:tvexp}, the total variation distance between the two processes is bounded above by the expectation of the total variation.
	\begin{align*}
		\norm{\L(X_{n+1})-\L(X'_{n+1})}&\leq E[\norm{\L((\beta_0+\beta_1 X_{n})Z_{n+1})-\L((\beta_0+\beta_1 X'_{n})Z_{n+1})}]
	\end{align*}
	Note that $Z_{n+1}$ and $Z'_{n+1}$ are used interchangeably in the total variation distance since $Z_{n+1}\overset{d}{=}Z'_{n+1}$. Let $Y_{n}=\beta_0+\beta_1 X_{n}$, $Y'_{n}=\beta_0+\beta_1 X'_{n}$, $\Delta=Y'_{n}-Y_{n}$, and $\Delta'=\frac{\Delta}{Y_{n}}$. WLOG $Y'_{n}>Y_{n}$ so that $\Delta, \Delta'>0$. Then,
	\begin{align*}
		\norm{\L(X_{n+1})-\L(X'_{n+1})}&\leq E[\norm{\L(Y_{n}Z_{n+1})-\L(Y'_{n}Z_{n+1})}]&\text{by Proposition \ref{prop:tvexp}}\\
		&=E[\norm{\L(Y_{n}Z_{n+1})-\L((Y_{n}+\Delta)Z_{n+1})}]\\
		&=E[\norm{\L(Z_{n+1})-\L((1+\Delta')Z_{n+1})}]&\text{by Proposition \ref{prop:invertibletv}}\\
		&=E[\norm{\L(\log(Z_{n+1}))-\L(\log(1+\Delta')+\log(Z_{n+1}))}]&\text{by Proposition \ref{prop:invertibletv}}\\
		&\leq \frac{M+1}{2}\sup_x e^x f_{Z_n}(e^x)E[\log(1+\Delta')]&\text{by lem \ref{lem:shiftedfunc3}. See prf of lem \ref{lem:coalescingcondsideways} for more details}\\
		&\leq \frac{M+1}{2}\sup_x e^x f_{Z_n}(e^x)\frac{E[\abs{\Delta}]}{\beta_0}&\text{by the Mean Value Theorem}\\
		&= \frac{M+1}{2}\sup_x e^x f_{Z_n}(e^x)\frac{\beta_1E[\abs{X_n-X'_n}]}{\beta_0}
	\end{align*}
\end{proof}

\subsubsection{Proof of lemmas used in Theorem \ref{thm:asymmetricarch}}

\begin{proof}[Proof of Lemma \ref{lem:asymarchcontraction}]\label{proof:asymarchcontraction}
	Let $\{X_n\}_{n\geq 1}\in \mathbb{R}$ and $\{X'_n\}_{n\geq 1}\in \mathbb{R}$ be two copies of the asymmetric ARCH process.
	
	For a fixed $n\geq 1$, let $Z_n=Z'_n$ and so,
	\begin{align*}
		E[\abs{X_n-X'_n}]&= E[\abs{\sqrt{(aX_{n-1}+b)^2+c^2}Z_n-\sqrt{(aX'_{n-1}+b)^2+c^2}Z_n}]\\
		&= \abs{\sqrt{(aX_{n-1}+b)^2+c^2}-\sqrt{(aX'_{n-1}+b)^2+c^2}}E[\abs{Z_n}]
	\end{align*}
	Note that the derivative of $f(x)=\sqrt{(ax+b)^2+c^2}$ is 
	\begin{equation}\label{eqn:asymarchcontraction}
		\abs{f'(x)}=\abs{\frac{a(ax+b)}{\sqrt{(ax+b)^2+c^2}}}\leq \frac{\abs{a(ax+b)}}{\sqrt{(ax+b)^2}}=\abs{a}
	\end{equation}
	and so,
	\begin{align*}
		E[\abs{X_n-X'_n}]&\leq \abs{a} E[\abs{Z_n}] E[\abs{X_{n-1}-X'_{n-1}}]
	\end{align*}
	Thus, the geometric convergence rate is $D=\abs{a} E[\abs{Z_0}]$.
\end{proof}

\begin{proof}[Proof of Lemma \ref{lem:asymarchcoalesc}]\label{proof:asymarchcoalesc}
	Let $\{X_n\}_{n\geq 1}\in \mathbb{R}$ and $\{X'_n\}_{n\geq 1}\in \mathbb{R}$ be two copies of the asymmetric ARCH process.
	
	For $n\geq 1$, $Z_n, Z'_n$ are independent. By Proposition \ref{prop:tvexp}, the total variation distance between the two processes is bounded above by the expectation of the total variation with respect to $X_{n-1},X'_{n-1}, Z_n, Z'_n$.
	\begin{align*}
		\norm{\L(X_n)-\L(X'_n)}&\leq E[\norm{\L(\sqrt{(aX_{n-1}+b)^2+c^2}Z_n)-\L(\sqrt{(aX'_{n-1}+b)^2+c^2}Z'_n)}]
	\end{align*}
	Let $Y_{n-1}=\sqrt{(aX_{n-1}+b)^2+c^2}$ and $Y'_{n-1}=\sqrt{(aX_{n-1}+b)^2+c^2}$, $\Delta=Y'_{n-1}-Y_{n-1}$ and $\Delta'=\frac{\Delta}{Y_{n-1}}$. WLOG, $Y'_{n-1}<Y_{n-1}$, so $-1< \Delta' <0$, because $Y_{n-1},Y'_{n-1}>0$ and
	\begin{align*}
		\norm{\L(X_n)-\L(X'_n)}&\leq E[\norm{\L(Y_{n-1}Z_n)-\L(Y'_{n-1}Z_n)}]\\
		&= E[\norm{\L(Y_{n-1}Z_n)-\L((Y_{n-1}+\Delta)Z_n)}]&\text{by Proposition \ref{prop:invertibletv}}\\
		&= E[\norm{\L(Z_n)-\L((1+\Delta')Z_n)}]&\text{by Proposition \ref{prop:invertibletv}}\\
		&\leq E\left[\sup_{x} 1-\frac{\pi_{Z_n}(x)}{\pi_{(1+\Delta')Z_n}(x)}\right]&\text{by Lemma 6.16 of \cite{markovmixing}}
	\end{align*}
	Let the density of $Z_n$ be $\pi_{Z_n}(x)$, then $\pi_{(1+\Delta')Z_n}(x)= \frac{1}{1+\Delta'}\pi_{Z_n}\left(\frac{x}{1+\Delta'}\right)$. 
	\begin{align*}
		\norm{\L(X_n)-\L(X'_n)} &\leq  E\left[\sup_{x} 1-(1+\Delta')\frac{\pi_{Z_n}(x)}{\pi_{Z_n}\left(\frac{x}{1+\Delta'}\right)}\right]\\
		&\leq E[\sup_{x} 1-(1+\Delta')]&\text{by assumption $\pi_{Z_n}(x)\geq \pi_{Z_n}\left(\frac{x}{1+\Delta'}\right)$}\\
		&= E[\Delta']\\
		&\leq \frac{E[\abs{Y_{n-1}-Y'_{n-1}}]}{c}&\text{since $Y_{n-1}\geq c$}\\
		&\leq \frac{\abs{a}}{c}E[\abs{X_{n-1}-X'_{n-1}}]&\text{by equation \ref{eqn:asymarchcontraction}}
	\end{align*}
\end{proof}

\subsubsection{Proof of lemmas used in Theorem \ref{thm:garch}}

\begin{proof}[Proof of Lemma \ref{lem:garchcontraction}]\label{proof:garchcontraction}
	Let $\{X_n\}_{n\geq 1}\in \mathbb{R}$ and $\{X'_n\}_{n\geq 1}\in \mathbb{R}$ be two copies of the GARCH process.
	For $n\geq 2$, let $Z_n=Z'_n$. First note that,
	\begin{align}\label{eq:garcheq}
		E[\abs{X_n-X'_n}]= E[\abs{\sigma_nZ_n -\sigma'_n Z_n}]= E[\abs{\sigma_n -\sigma'_n} \abs{Z_n}]=E[\abs{\sigma_n -\sigma'_n}]E[ \abs{Z_n}]
	\end{align}
	
	Next, we find an upper bound on $E[\abs{\sigma_n-\sigma'_n}]$ by first noting that $\sigma^2_n=\alpha^2+(\beta^2 Z^2_{n-1}+\gamma^2)\sigma^2_{n-1}$ by substitution. 
	\begin{align*}
		E[\abs{\sigma_n-\sigma'_n}] &= E[\abs{\sqrt{\alpha^2+(\beta^2 Z^2_{n-1}+\gamma^2)\sigma^2_{n-1}}-\sqrt{\alpha^2+(\beta^2 Z^2_{n-1}+\gamma^2)\sigma^{'2}_{n-1}}}] \\
		&\leq E[\sqrt{\beta^2 Z^2_{n-1}+\gamma^2}] E[\abs{\sigma_{n-1}-\sigma^{'}_{n-1}}] &\text{taking max of the derivative}\\
		&=E[\sqrt{\beta^2 Z^2_{n-1}+\gamma^2}] \frac{E[\abs{X_{n-1}-X'_{n-1}}]}{E[\abs{Z_{n-1}}]} &\text{by equation \ref{eq:garcheq}}
	\end{align*}
	Finally, substituting $E[\abs{\sigma_n-\sigma'_n}]$ into equation \ref{eq:garcheq},
	\begin{align*}
		E[\abs{X_n-X'_n}]&\leq E[\sqrt{\beta^2 Z^2_{n-1}+\gamma^2}] \frac{E[\abs{X_{n-1}-X'_{n-1}}]}{E[\abs{Z_{n-1}}]}E[\abs{Z_n}]\\
		&= E[\sqrt{\beta^2 Z^2_{n-1}+\gamma^2}] E[\abs{X_{n-1}-X'_{n-1}}]\\
		&\leq \sqrt{\beta^2 E[Z_0^2]+\gamma^2} E[\abs{X_{n-1}-X'_{n-1}}] &\text{by Jensen's inequality}
	\end{align*}
	Thus, the geometric convergence rate is $D=\sqrt{\beta^2 E[Z_0^2]+\gamma^2}$.
\end{proof}

\begin{proof}[Proof of Lemma \ref{lem:garchcoalescing}]\label{proof:garchcoalescing}
	Let $\{X_n\}_{n\geq 1}\in \mathbb{R}$ and $\{X'_n\}_{n\geq 1}\in \mathbb{R}$ be two copies of the GARCH process.
	
	For $n\geq 2$, suppose that $Z_n, Z'_n$ are independent. By Proposition \ref{prop:tvexp}, the total variation distance between the two processes is bounded above by the expectation of the total variation.
	\begin{align*}
		\norm{\L(X_n)-\L(X'_n)}&\leq E[\norm{\L(\sigma_nZ_n)-\L(\sigma'_n Z_n)}]
	\end{align*}
	Let $\Delta=\sigma'_{n}-\sigma_{n}$ and $\Delta'=\frac{\Delta}{\sigma_{n}}$. WLOG, $\sigma'_{n}<\sigma_{n}$, so $\Delta, \Delta' <0$ because $\sigma_{n},\sigma'_{n}>0$ and
	\begin{align*}
		\norm{\L(X_n)-\L(X'_n)}&= E[\norm{\L(\sigma_{n}Z_n)-\L((\sigma_{n}+\Delta)Z_n)}]&\text{by Proposition \ref{prop:invertibletv}}\\
		&= E[\norm{\L(Z_n)-\L((1+\Delta')Z_n)}]&\text{by Proposition \ref{prop:invertibletv}}\\
		&\leq E\left[\sup_{x} 1-\frac{\pi_{Z_n}(x)}{\pi_{(1+\Delta')Z_n}(x)}\right] &\text{by Lemma 6.16 of \cite{markovmixing}}
	\end{align*}
	Let the density of $Z_n$ be $\pi_{Z_n}(x)$, then $\pi_{(1+\Delta')Z_n}(x)= \frac{1}{1+\Delta'}\pi_{Z_n}\left(\frac{x}{1+\Delta'}\right)$. 
	\begin{align*}
		\norm{\L(X_n)-\L(X'_n)} &\leq  E\left[\sup_{x} 1-(1+\Delta')\frac{\pi_{Z_n}(x)}{\pi_{Z_n}\left(\frac{x}{1+\Delta'}\right)}\right]\\
		&\leq E[\sup_{x} 1-(1+\Delta')]&\text{by assumption $\pi_{Z_n}(x)\geq \pi_{Z_n}\left(\frac{x}{1+\Delta'}\right)$}\\
		&= E[\Delta']\\
		&\leq \frac{E[\abs{\sigma'_{n}-\sigma_{n}}]}{\alpha}&\text{since $\sigma_{n}\geq \alpha$}\\
		&\leq \frac{D}{\alpha E[\abs{Z_{n-1}}]} E[\abs{X_{n-1}-X'_{n-1}}] & \text{by equation in proof \ref{proof:garchcontraction}}
	\end{align*}
\end{proof}

\begin{proof}[Proof of Lemma \ref{lem:garchinitial}]\label{proof:garchinitial}
	\begin{align*}
		E[\abs{X_1 - X'_1}] &= \abs{\sigma^{2}_1 - \sigma^{'2}_1} E[\abs{Z_1}] &\text{by equation in proof \ref{proof:garchcontraction}}\\
		&= \abs{\sqrt{\alpha^2 +\beta^2 X_0^{2} + \gamma^2 \sigma_0^{2}} - \sqrt{\alpha^2 +\beta^2 X_0^{'2} + \gamma^2 \sigma_0^{'2}}} E[\abs{Z_1}] \\
		&\leq \sqrt{\abs{(\alpha^2 +\beta^2 X_0^{2} + \gamma^2 \sigma_0^{2}) - (\alpha^2 +\beta^2 X_0^{'2} + \gamma^2 \sigma_0^{'2})}} E[\abs{Z_1}] \\
		& &\hspace{-6cm} \text{since $\abs{\sqrt{x}-\sqrt{y}}=\sqrt{(\sqrt{x}-\sqrt{y})^2} = \sqrt{x+y-2\sqrt{x}\sqrt{y}} \leq \sqrt{\abs{x-y}}$} \\
		& \leq \sqrt{\beta^2 \abs{X_0^{2}-X_0^{'2}} + \gamma^2 \abs{\sigma_0^{2} - \sigma_0^{'2}}} E[\abs{Z_0}]
	\end{align*}
\end{proof}

\end{document}